\documentclass[submission,copyright,creativecommons]{eptcs}
\usepackage[utf8x]{inputenc}
\providecommand{\event}{QPL 2020} 

\usepackage{graphicx}
\usepackage{subcaption}
\usepackage{dcolumn}
\usepackage{bm}
\usepackage{tikz-cd}
\usepackage{tikz}
\usepackage{bbm}
\usepackage{physics}
\usepackage{dsfont}
\usepackage{float}
\usepackage{amsthm}
\usepackage{circuitikz}
\usepackage{tikzit}
\usepackage{url}

\usepackage{amsmath}
\usepackage{amsfonts}
\usepackage{amssymb}
\usepackage{calrsfs} 
\DeclareMathAlphabet{\pazocal}{OMS}{zplm}{m}{n}

\tikzstyle{doubled}=[line width=1.5pt] 

\tikzstyle{dot}=[inner sep=0mm,minimum width=2mm,minimum height=2mm,draw,shape=circle]  
\tikzstyle{ddot}=[inner sep=0mm, doubled, minimum width=2.5mm,minimum height=2.5mm,draw,shape=circle]

\tikzstyle{pdot}=[inner sep=0mm, doubled, minimum width=2.5mm,minimum height=2.5mm,shape=circle]
\tikzstyle{phase dimensions}=[minimum size=6mm,font=\footnotesize,inner sep=0.2mm,outer sep=-2mm]

\tikzstyle{phase dot}=[pdot,phase dimensions]
\tikzstyle{wphase dot}=[dot, phase dimensions]

\tikzstyle{hadamard}=[fill=white,draw,inner sep=0.6mm,font=\footnotesize,minimum height=4mm,minimum width=4mm]

\tikzstyle{anti} = [shade, bottom color=black, top color = white, draw, minimum height = 4mm, minimum width = 4mm]

\tikzstyle{triang}=[regular polygon,regular polygon sides=3,draw,scale=0.75,inner sep=-0.5pt,minimum width=9mm,fill=white,regular polygon rotate=180]
\tikzstyle{triangdag}=[regular polygon,regular polygon sides=3,draw,scale=0.75,inner sep=-0.5pt,minimum width=9mm,fill=white]

\makeatletter
\newcommand{\boxshape}[3]{%
\pgfdeclareshape{#1}{
\inheritsavedanchors[from=rectangle] 
\inheritanchorborder[from=rectangle]
\inheritanchor[from=rectangle]{center}
\inheritanchor[from=rectangle]{north}
\inheritanchor[from=rectangle]{south}
\inheritanchor[from=rectangle]{west}
\inheritanchor[from=rectangle]{east}
\backgroundpath{
\southwest \pgf@xa=\pgf@x \pgf@ya=\pgf@y
\northeast \pgf@xb=\pgf@x \pgf@yb=\pgf@y

\@tempdima=#2
\@tempdimb=#3

\pgfpathmoveto{\pgfpoint{\pgf@xa - 5pt + \@tempdima}{\pgf@ya}}
\pgfpathlineto{\pgfpoint{\pgf@xa - 5pt - \@tempdima}{\pgf@yb}}
\pgfpathlineto{\pgfpoint{\pgf@xb + 5pt + \@tempdimb}{\pgf@yb}}
\pgfpathlineto{\pgfpoint{\pgf@xb + 5pt - \@tempdimb}{\pgf@ya}}
\pgfpathlineto{\pgfpoint{\pgf@xa - 5pt + \@tempdima}{\pgf@ya}}
\pgfpathclose
}
}}

\boxshape{NEbox}{0pt}{5pt}
\boxshape{SEbox}{0pt}{-5pt}
\boxshape{NWbox}{5pt}{0pt}
\boxshape{SWbox}{-5pt}{0pt}
\boxshape{EBox}{-3pt}{3pt}
\boxshape{WBox}{3pt}{-3pt}
\makeatother

\tikzstyle{map}=[draw,shape=SEbox,inner sep=2pt,minimum height=6mm,fill=white]
\tikzstyle{mapdag}=[draw,shape=NEbox,inner sep=2pt,minimum height=6mm,fill=white]
\tikzstyle{maptrans}=[draw,shape=NWbox,inner sep=2pt,minimum height=6mm,fill=white]
\tikzstyle{mapconj}=[draw,shape=SWbox,inner sep=2pt,minimum height=6mm,fill=white]

\tikzstyle{dmap}=[draw,doubled,shape=SEbox,inner sep=2pt,minimum height=6mm,fill=white]
\tikzstyle{dmapdag}=[draw,doubled,shape=NEbox,inner sep=2pt,minimum height=6mm,fill=white]
\tikzstyle{dmaptrans}=[draw,doubled,shape=NWbox,inner sep=2pt,minimum height=6mm,fill=white]
\tikzstyle{dmapconj}=[draw,doubled,shape=SWbox,inner sep=2pt,minimum height=6mm,fill=white]

\makeatletter
\pgfdeclareshape{cornerpoint}{
\inheritsavedanchors[from=rectangle] 
\inheritanchorborder[from=rectangle]
\inheritanchor[from=rectangle]{center}
\inheritanchor[from=rectangle]{north}
\inheritanchor[from=rectangle]{south}
\inheritanchor[from=rectangle]{west}
\inheritanchor[from=rectangle]{east}
\backgroundpath{
\southwest \pgf@xa=\pgf@x \pgf@ya=\pgf@y
\northeast \pgf@xb=\pgf@x \pgf@yb=\pgf@y

\pgfmathsetmacro{\pgf@shorten@left}{\pgfkeysvalueof{/tikz/shorten left}}
\pgfmathsetmacro{\pgf@shorten@right}{\pgfkeysvalueof{/tikz/shorten right}}

\pgfpathmoveto{\pgfpoint{0.5 * (\pgf@xa + \pgf@xb)}{\pgf@ya - 5pt}}
\pgfpathlineto{\pgfpoint{\pgf@xa - 8pt + \pgf@shorten@left}{\pgf@yb - 1.5 * \pgf@shorten@left}}
\pgfpathlineto{\pgfpoint{\pgf@xa - 8pt + \pgf@shorten@left}{\pgf@yb}}
\pgfpathlineto{\pgfpoint{\pgf@xb + 8pt - \pgf@shorten@right}{\pgf@yb}}
\pgfpathlineto{\pgfpoint{\pgf@xb + 8pt - \pgf@shorten@right}{\pgf@yb - 1.5 * \pgf@shorten@right}}
\pgfpathclose
}
}

\pgfdeclareshape{cornercopoint}{
\inheritsavedanchors[from=rectangle] 
\inheritanchorborder[from=rectangle]
\inheritanchor[from=rectangle]{center}
\inheritanchor[from=rectangle]{north}
\inheritanchor[from=rectangle]{south}
\inheritanchor[from=rectangle]{west}
\inheritanchor[from=rectangle]{east}
\backgroundpath{
\southwest \pgf@xa=\pgf@x \pgf@ya=\pgf@y
\northeast \pgf@xb=\pgf@x \pgf@yb=\pgf@y

\pgfmathsetmacro{\pgf@shorten@left}{\pgfkeysvalueof{/tikz/shorten left}}
\pgfmathsetmacro{\pgf@shorten@right}{\pgfkeysvalueof{/tikz/shorten right}}

\pgfpathmoveto{\pgfpoint{0.5 * (\pgf@xa + \pgf@xb)}{\pgf@yb + 5pt}}
\pgfpathlineto{\pgfpoint{\pgf@xa - 8pt + \pgf@shorten@left}{\pgf@ya + 1.5 * \pgf@shorten@left}}
\pgfpathlineto{\pgfpoint{\pgf@xa - 8pt + \pgf@shorten@left}{\pgf@ya}}
\pgfpathlineto{\pgfpoint{\pgf@xb + 8pt - \pgf@shorten@right}{\pgf@ya}}
\pgfpathlineto{\pgfpoint{\pgf@xb + 8pt - \pgf@shorten@right}{\pgf@ya + 1.5 * \pgf@shorten@right}}
\pgfpathclose
}
}

\makeatother

\pgfkeyssetvalue{/tikz/shorten left}{0pt}
\pgfkeyssetvalue{/tikz/shorten right}{0pt}

\tikzstyle{kpoint common}=[draw,fill=white,inner sep=1pt,minimum height=4mm]
\tikzstyle{kpoint}=[shape=cornerpoint,shorten left=5pt,kpoint common]
\tikzstyle{kpoint adjoint}=[shape=cornercopoint,shorten left=5pt,kpoint common]
\tikzstyle{kpoint conjugate}=[shape=cornerpoint,shorten right=5pt,kpoint common]
\tikzstyle{kpoint transpose}=[shape=cornercopoint,shorten right=5pt,kpoint common]

\tikzstyle{kpointdag}=[kpoint adjoint]
\tikzstyle{kpointadj}=[kpoint adjoint]
\tikzstyle{kpointconj}=[kpoint conjugate]
\tikzstyle{kpointtrans}=[kpoint transpose]

\tikzstyle{big kpoint}=[kpoint, minimum width=1.2 cm, minimum height=8mm, inner sep=4pt, text depth=3mm]

 \tikzstyle{upground}=[circuit ee IEC,thick,ground,rotate=90,scale=1.5]
 \tikzstyle{downground}=[circuit ee IEC,thick,ground,rotate=-90,scale=1.5]

\tikzstyle{discarding}=[fill=white, draw=black, shape=circle, style=upground]
\tikzstyle{smalldiscarding}=[fill=white, draw=black, shape=circle, style=upground, scale=0.5]
\tikzstyle{backdiscard}=[fill=white, draw=black, shape=circle, style=downground]
\tikzstyle{smallbackdiscard}=[fill=white, draw=black, shape=circle, style=downground, scale=0.5]
\tikzstyle{state}=[fill=white, draw=black, style=triang, tikzit shape=rectangle]
\tikzstyle{kstate}=[fill=white, draw=black, style=kpoint, tikzit shape=rectangle]
\tikzstyle{kstateconj}=[fill=white, draw=black, style=kpoint conjugate, tikzit shape=rectangle]
\tikzstyle{kstateBIG}=[fill=white, draw=black, style=big kpoint, tikzit shape=rectangle]
\tikzstyle{effect}=[fill=white, draw=black, style=triangdag]
\tikzstyle{keffect}=[fill=white, draw=black, style=kpoint adjoint]
\tikzstyle{keffectconj}=[fill=white, draw=black, style=kpoint transpose]
\tikzstyle{morphdag}=[style=mapdag]
\tikzstyle{morph}=[style=map]
\tikzstyle{morphtrans}=[style=maptrans]
\tikzstyle{morphconj}=[style=mapconj]
\tikzstyle{CPMmorph}=[style=dmap]
\tikzstyle{CPMmorphconj}=[style=dmapconj]
\tikzstyle{CPMkstate}=[fill=white, draw=black, style=kpoint, tikzit shape=rectangle, doubled]
\tikzstyle{CPMkstateconj}=[fill=white, draw=black, style=kpoint conjugate, tikzit shape=rectangle, doubled]
\tikzstyle{CPMkstateBIG}=[fill=white, draw=black, style=big kpoint, tikzit shape=rectangle, doubled]
\tikzstyle{CPMkeffect}=[fill=white, draw=black, style=kpoint adjoint, doubled]
\tikzstyle{CPMkeffectconj}=[fill=white, draw=black, style=kpoint transpose, doubled]
\tikzstyle{leak}=[style=tinypoint, regular polygon rotate=-90]
\tikzstyle{leakfill}=[style=tinypoint, regular polygon rotate=-90, fill=black]
\tikzstyle{Z}=[style=dot, fill=green]
\tikzstyle{X}=[style=dot, fill=red]
\tikzstyle{black_dot}=[style=dot, fill=black]
\tikzstyle{white_dot}=[style=dot, fill=white]
\tikzstyle{qblack_dot}=[style=ddot, fill=black]
\tikzstyle{qwhite_dot}=[style=ddot, fill=white]
\tikzstyle{whitephase}=[style=wphase dot, fill=white]
\tikzstyle{qredphase}=[style=phase dot, fill=red]
\tikzstyle{qgreenphase}=[style=phase dot, fill=green]
\tikzstyle{had}=[style=hadamard, doubled]
\tikzstyle{classhad}=[style=hadamard]

\tikzstyle{dottededge}=[-, dotted]
\tikzstyle{double edge}=[-, style=doubled, draw=black, tikzit draw={rgb,255: red,191; green,0; blue,64}]
\tikzstyle{Red}=[-, draw={rgb,255: red,239; green,15; blue,45}]
\tikzstyle{Pink}=[-, draw={rgb,255: red,248; green,7; blue,224}]
\tikzstyle{Blue}=[-, draw={rgb,255: red,45; green,4; blue,249}]
\tikzstyle{Turq}=[-, draw=black]
\tikzstyle{Green}=[-, draw={rgb,255: red,4; green,241; blue,20}]
\tikzstyle{Yellow}=[-, draw={rgb,255: red,244; green,244; blue,5}]
\tikzstyle{Beige}=[-, draw={rgb,255: red,200; green,221; blue,143}]
\tikzstyle{Grey}=[-, draw={rgb,255: red,144; green,144; blue,144}]
\tikzstyle{new edge style 0}=[->]

\usetikzlibrary{calc, positioning, shapes.geometric}
\usetikzlibrary{
	arrows,
	shapes,
	decorations,
	intersections,
	backgrounds,
	positioning,
	circuits.ee.IEC
	}

\theoremstyle{definition}

\theoremstyle{plain}
\newtheorem{thm}{Theorem}
\theoremstyle{plain}

\theoremstyle{plain}

\newtheorem{coro}[thm]{Corollary}

\title{A Diagrammatic Approach to Information Transmission in Generalised Switches}
\author{Matt Wilson \institute{Department of Computer Science, University of Oxford, Wolfson Building, Parks Road, Oxford, UK} \institute{HKU-Oxford Joint Laboratory for Quantum Information and Computation} \email{matthew.wilson@cs.ox.ac.uk} \and Giulio Chiribella
\institute{QICI Quantum Information and Computation Initiative, Department of Computer Science}
\institute{Department of Computer Science, University of Oxford, Wolfson Building, Parks Road, Oxford, UK} 
\institute{HKU-Oxford Joint Laboratory for Quantum Information and Computation}
\institute{Perimeter Institute for Theoretical Physics, 31 Caroline Street North, Waterloo, Ontario, Canada}
\email{giulio.chiribella@cs.ox.ac.uk}}

\def\titlerunning{A Diagrammatic Approach to Information Transmission in Generalised Switches}
\def\authorrunning{M. Wilson \& G. Chiribella}

\begin{document}

\maketitle

\begin{abstract}
The quantum switch is a higher-order operation that takes as an input two quantum processes and combines them in a coherent superposition of two alternative orders.   Here we provide an approach to the quantum switch based on the methods of categorical quantum mechanics. Specifically,  we represent the quantum switch as a sum of diagrams in the category of finite dimensional Hilbert spaces, or, equivalently, as a sum of diagrams built from Selinger's CPM construction.  The sum-of-diagrams picture provides intuition for the activation of classical capacity of completely depolarising channels (CDPCs) and allows for generalisation to $N$-channel switches. We demonstrate the use of these partially diagrammatic methods by deriving a permutation condition for computing the output of any $N$-channel switch of CDPCs, we then use that condition to prove that amongst all possible terms, the interference terms associated to cyclic permutations of the $N$ channels are the information-transmitting terms with maximum normalisation.\end{abstract}

\section{Introduction}
Quantum Shannon theory \cite{Schumacher1995QuantumCoding,Wilde2013QuantumTheory} explores the extension of Shannon's information theory to scenarios where the information carriers are quantum systems.   Recently, there has been an interest in a further extension, where  not only the information carriers, but  also the configuration of the communication channels, can be quantum \cite{Ebler2018EnhancedOrder,Chiribella2018IndefiniteChannel,Salek2018QuantumOrders,Abbott2018CommunicationChannels,Guerin2019CommunicationNoise,Kristjansson2019ResourceProcesses}. These extensions allow the communication channels to be combined in more general ways than those allowed in standard quantum Shannon theory. Technically, the combination of channels is described by a quantum supermap  \cite{Chiribella2008TransformingSupermaps}, a higher-order transformation that maps channels into channels.  A paradigmatic example of supermap is the quantum switch \cite{chiribella2009beyond, Chiribella2013QuantumStructure} of two channels $\pazocal{N}^{(1)}$, $\pazocal{N}^{(2)}$ which superposes the two possible sequential compositions $\pazocal{N}^{(1)} \circ \pazocal{N}^{(2)}$ and $\pazocal{N}^{(2)} \circ \pazocal{N}^{(1)}$.
The quantum switch has been shown to offer  computational advantages \cite{Chiribella2012PerfectStructures,Oreshkov2012QuantumOrder,Araujo2014ComputationalGates,Guerin2016ExponentialCommunication}, as well as advantages in quantum metrology  \cite{zhao2019quantum,Mukhopadhyay2018SuperpositionThermometry}. 

In this paper we analyze  the quantum switch from the perspective of Categorical Quantum Mechanics (CQM) \cite{Abramsky2004AProtocols, Coecke2010QuantumPicturalism} which has been used to analyse a variety of quantum protocols \cite{Boixo2012EntangledDephasing,Chancellor2016GraphicalCorrection, Gogioso2017FullyProblem,Kissinger2017Picture-perfectDistribution,Ranchin2014CompleteMechanics,Vicary2012TheAlgorithms,Zeng2014AbstractAlgorithms}, and to formalise causality \cite{Coecke2013CausalProcesses} and causal structure \cite{Kissinger2019AStructure}. Here we use the diagrammatic language of CQM to analyse information processing advantages of the quantum switch by writing each output as a sum of diagrams built from the CPM-construction \cite{Selinger2007DaggerAbstract, Coecke2008AxiomaticCPM-construction}. 
We focus specifically on the  activation of capacities of completely depolarising channels (CDPCs) shown  in \cite{Ebler2018EnhancedOrder}: there, it was shown that two CDPCs acting in a superposition of two alternative orders enable the transmission of classical information, even though each individual channel blocks information entirely.    In the following, we will consider the generalisation to $N\ge 2$  CDPCs,  with quantum superpositions   of general permutations of $N$ completely depolarising channels.   

Previous work on this subject was done by Procopio {\em et al} in Ref. \cite{Procopio2019CommunicationScenario}, where a formula is given for the action of the quantum switch of $N$ partially depolarising channels $\{\pazocal{N}^{(i)}\}_{i=1}^{N}$. Here we provide  an explicit expression for the output of such an $N$-channel switch using only diagrammatic manipulations based on the algebra of permutations. This permutation condition is then used as a heuristic to suggest that maximum capacity enhancements may be associated to the switch of the $N$ cyclic permutations.   The communication enhancements resulting from the superposition of cyclic permutations are independently discussed in \cite{chiribella2020quantumcomm} and \cite{sazim2020classical}.



\section{Preliminaries}\label{sec:switch}
We first review the algebraic representation of  quantum channels and  of  the quantum switch within the Hilbert space framework of quantum mechanics.  Then, we review the diagrammatic representation of quantum channels in the  framework of categorical quantum mechanics.
\subsection{Algebraic Presentation of a Quantum Channel}
In the  pure state picture of quantum mechanics, a quantum state is represented by a normalised element $\ket{\psi}$ of a Hilbert space $\cal H$, up to a global phase. Pure quantum states are then  generalised to mixed states, described by unit-trace positive linear operators $\rho \in L(\cal H)$, $L(\cal H)$ denoting the set of linear operators on the Hilbert space  $\cal H$. 

A quantum channel $\pazocal{N}: L({\cal H}) \rightarrow L(\cal H)$ is any transformation $\rho \mapsto \pazocal{N}(\rho)$ which is linear, trace preserving, and completely positive. For any quantum channel there exists operators $\{K_{i}\}$ such that 
  $  \pazocal{N}(\rho) = \sum_{i} K_{i}\rho K^{\dagger}_{i}$ $\forall \rho  \in L(\cal H)$, 
which is referred to as a Kraus decomposition of a channel $\pazocal{N}$ into Kraus operators $\{K_{i}\}$. Two canonical examples of quantum channels  are the identity channel $\pazocal{I}$ and the completely depolarising channel $\pazocal{D}$, defined as 
$\pazocal{I}(\rho) = \rho$ and $\pazocal{D}(\rho) = \frac{I}{d}$, respectively.
\subsection{Algebraic Presentation of the Quantum Switch}

Channels are transformations of states. One can also consider transformations of channels,  an idea formally captured by the framework of quantum supermaps \cite{Chiribella2008TransformingSupermaps,chiribella2009theoretical,Chiribella2013QuantumStructure}. The quantum switch \cite{chiribella2009beyond,Chiribella2013QuantumStructure} is a bipartite supermap, from a pair of channels $\pazocal{N}^{(1)}$, $\pazocal{N}^{(2)}$ and a fixed control qubit $\ket{+}\bra{+}$ it produces a superposition of sequential compositions $\pazocal{N}^{(1)}\circ \pazocal{N}^{(2)}$ and $\pazocal{N}^{(2)}\circ \pazocal{N}^{(1)}$. 

The quantum switch $S$ of $\pazocal{N}^{(1)}$ and $\pazocal{N}^{(2)}$, with Kraus decompositions $\{ K^{(1)}_{i_{1}} \}_{i_{1} = 1}^{d^2}$ and $\{ K^{(2)}_{i_{2}} \}_{i_{2} =1}^{d^2}$ respectively can be split into four components
\begin{align}
    S(\pazocal{N}^{(1)}, \pazocal{N}^{(2)})(\rho,\ket{+}\bra{+}) = \frac{1}{2}   \sum_{i_{1}}^{d^2}\sum_{i_{2}}^{d^2} K_{i_{1}}^{(1)}  K_{i_{2}}^{{(2)}}\rho K_{i_{2}}^{{(2)}^{\dagger}} K_{i_{1}}^{{(1)}^{\dagger}} \otimes \ket{0}\bra{0} \nonumber \\
    + \frac{1}{2}   \sum_{i_{1}}^{d^2}\sum_{i_{2}}^{d^2} K_{i_{2}}^{(2)}  K_{i_{1}}^{{(1)}}\rho K_{i_{1}}^{{(1)}^{\dagger}} K_{i_{2}}^{{(2)}^{\dagger}} \otimes \ket{1}\bra{1} \nonumber \\
    +\frac{1}{2}   \sum_{i_{1}}^{d^2}\sum_{i_{2}}^{d^2} K_{i_{1}}^{(1)}  K_{i_{2}}^{{(2)}}\rho K_{i_{1}}^{{(1)}^{\dagger}} K_{i_{2}}^{{(2)}^{\dagger}} \otimes \ket{0}\bra{1} \nonumber \\
    +\frac{1}{2}   \sum_{i_{1}}^{d^2}\sum_{i_{2}}^{d^2} K_{i_{2}}^{(2)}  K_{i_{1}}^{{(1)}}\rho K_{i_{2}}^{{(2)}^{\dagger}} K_{i_{1}}^{{(1)}^{\dagger}} \otimes \ket{1}\bra{0}
\end{align}
Interference between the two sequential orderings of $\pazocal{N}^{(1)}$ and $\pazocal{N}^{(2)}$ is seen in the off diagonal elements of the control qubit. When $\pazocal{N}^{(1)}$ and $\pazocal{N}^{(2)}$ are both completely depolarising channels (CDPCs), $\pazocal{N}^{(1)} = \pazocal{N}^{(2)} = \pazocal{D}$, 
the algebraic properties of the Kraus decomposition of $\pazocal{D}$ can be used to compute the output explicitly \cite{Ebler2018EnhancedOrder}.
\begin{equation}\label{eq2}
    S(\pazocal{N}^{(1)}, \pazocal{N}^{(2)})(\rho,\ket{+}\bra{+}) = \frac{1}{2} \sum_{i,j \in \{0,1\}} \left[ \delta_{ij} \frac{I}{d} + (1- \delta_{ij})\frac{\rho}{d^2} \right] \otimes \ket{i}\bra{j}
\end{equation}
The dependence of the output on $\rho$, implies this channel can transmit information, formally it has non-zero classical capacity \cite{Ebler2018EnhancedOrder}. The output after discarding of the control qubit is maximally mixed, information can only reach the receiver if the receiver additionally have access to the control system.

\subsection{Diagrammatic Representation of a Quantum Channel}
In the language of categorical quantum mechanics, Hilbert spaces are drawn as wires \cite{Coecke2010QuantumPicturalism}, 
\begin{equation}
    \begin{tikzpicture}
	\begin{pgfonlayer}{nodelayer}
		\node [style=none] (0) at (0, 0.5) {};
		\node [style=none] (1) at (0, -0.5) {};
		\node [style=none] (2) at (0.25, 0) {$H$};
	\end{pgfonlayer}
	\begin{pgfonlayer}{edgelayer}
		\draw (0.center) to (1.center);
	\end{pgfonlayer}
\end{tikzpicture}

\end{equation}
A density matrix $\rho$ is a box with output wires (wires pointing upwards), a quantum channel $\pazocal{N}$ has input and output wires.
\begin{equation}
    \begin{tikzpicture}
	\begin{pgfonlayer}{nodelayer}
		\node [style=none] (0) at (4.575, 0.25) {};
		\node [style=none] (1) at (4.575, -0.25) {};
		\node [style=none] (3) at (4.775, 0.25) {};
		\node [style=none] (4) at (4.8, -0.25) {};
		\node [style=none] (5) at (4.575, 1.25) {};
		\node [style=none] (6) at (4.575, 0.75) {};
		\node [style=none] (7) at (4.775, 1.25) {};
		\node [style=none] (8) at (4.775, 0.75) {};
		\node [style=morph] (9) at (4.75, 0.5) {$\pazocal{N}$};
		\node [style=none] (11) at (4.6, -0.25) {};
		\node [style=none] (13) at (4.75, -0.25) {};
		\node [style=none] (14) at (4.2, -0.25) {};
		\node [style=none] (15) at (5.2, -0.25) {};
		\node [style=none] (16) at (4.7, -1) {};
		\node [style=none] (17) at (4.7, -0.5) {$\rho$};
		\node [style=none] (18) at (3.5, 0.25) {$\equiv$};
		\node [style=none] (19) at (2, 0.25) {$\pazocal{N}(\rho)$};
		\node [style=none] (20) at (-0.4, 0.75) {};
		\node [style=none] (21) at (-0.4, -0.25) {};
		\node [style=none] (22) at (-0.25, 0.75) {};
		\node [style=none] (23) at (-0.25, -0.25) {};
		\node [style=none] (24) at (-1.05, 0.25) {$\equiv$};
		\node [style=none] (26) at (-2.35, 0.75) {};
		\node [style=kstateconj] (27) at (-2.35, -0.25) {$i$};
		\node [style=none] (28) at (-1.7, 0.75) {};
		\node [style=kstate] (29) at (-1.7, -0.25) {$j$};
		\node [style=none] (30) at (-4.25, 0.25) {$\approx$};
		\node [style=none] (31) at (-5.5, 0.25) {$\sum_{ij} \alpha_{ij} \ket{i}\bra{j}$};
		\node [style=none] (32) at (-3.3, 0.25) {$\sum_{ij} \alpha_{ij}$};
		\node [style=none] (33) at (-0.8, -0.25) {};
		\node [style=none] (34) at (0.2, -0.25) {};
		\node [style=none] (35) at (-0.3, -1) {};
		\node [style=none] (36) at (-0.3, -0.5) {$\rho$};
	\end{pgfonlayer}
	\begin{pgfonlayer}{edgelayer}
		\draw (0.center) to (1.center);
		\draw (3.center) to (4.center);
		\draw (5.center) to (6.center);
		\draw (7.center) to (8.center);
		\draw (14.center) to (16.center);
		\draw (16.center) to (15.center);
		\draw (15.center) to (14.center);
		\draw (20.center) to (21.center);
		\draw (22.center) to (23.center);
		\draw (26.center) to (27);
		\draw (28.center) to (29);
		\draw (33.center) to (35.center);
		\draw (35.center) to (34.center);
		\draw (34.center) to (33.center);
	\end{pgfonlayer}
\end{tikzpicture}
\end{equation}
Any plain wire can be considered an identity map and expanded as a resolution of the identity, similarly for a bent wire (or ``cap") representing the trace.
\begin{equation}
    \begin{tikzpicture}
	\begin{pgfonlayer}{nodelayer}
		\node [style=none] (0) at (-0.25, -0.25) {};
		\node [style=none] (1) at (1, -0.25) {};
		\node [style=none] (2) at (1.75, 0) {=};
		\node [style=none] (4) at (3.25, -0.25) {};
		\node [style=keffectconj] (6) at (3.25, 0.25) {$i$};
		\node [style=none] (7) at (2.5, 0.25) {$\sum_i$};
		\node [style=none] (8) at (4, -0.25) {};
		\node [style=keffect] (9) at (4, 0.25) {$i$};
		\node [style=none] (10) at (-4.25, 0.5) {};
		\node [style=none] (11) at (-4.25, -0.5) {};
		\node [style=none] (12) at (-3.75, 0) {=};
		\node [style=none] (13) at (-2.5, 1) {};
		\node [style=none] (14) at (-2.5, -1) {};
		\node [style=kstateconj] (15) at (-2.5, 0.5) {$i$};
		\node [style=keffectconj] (16) at (-2.5, -0.5) {$i$};
		\node [style=none] (17) at (-3, 0) {$\sum_i$};
	\end{pgfonlayer}
	\begin{pgfonlayer}{edgelayer}
		\draw [bend left=90, looseness=1.50] (0.center) to (1.center);
		\draw (4.center) to (6);
		\draw (8.center) to (9);
		\draw (10.center) to (11.center);
		\draw (14.center) to (16);
		\draw (15) to (13.center);
	\end{pgfonlayer}
\end{tikzpicture}
\end{equation}
A closed loop is then the dimension of the Hilbert space $d = \textrm{dim}(\cal H)$,
\begin{equation}
    \begin{tikzpicture}
	\begin{pgfonlayer}{nodelayer}
		\node [style=none] (0) at (-2.75, -0.25) {};
		\node [style=none] (1) at (-1.5, -0.25) {};
		\node [style=none] (2) at (-0.75, -0.25) {=};
		\node [style=none] (4) at (0.75, -0.25) {};
		\node [style=keffectconj] (6) at (0.75, 0.25) {$i$};
		\node [style=none] (7) at (0, 0.25) {$\sum_i$};
		\node [style=none] (8) at (1.5, -0.25) {};
		\node [style=keffect] (9) at (1.5, 0.25) {$i$};
		\node [style=none] (10) at (-2.75, -0.25) {};
		\node [style=none] (11) at (-1.5, -0.25) {};
		\node [style=none] (13) at (0.75, -0.25) {};
		\node [style=kstateconj] (14) at (0.75, -0.75) {$j$};
		\node [style=none] (15) at (0, -0.75) {$\sum_j$};
		\node [style=none] (16) at (1.5, -0.25) {};
		\node [style=kstate] (17) at (1.5, -0.75) {$j$};
		\node [style=none] (18) at (2.5, -0.25) {=};
		\node [style=none] (19) at (3.25, -0.25) {$d$};
	\end{pgfonlayer}
	\begin{pgfonlayer}{edgelayer}
		\draw [bend left=90, looseness=1.50] (0.center) to (1.center);
		\draw (4.center) to (6);
		\draw (8.center) to (9);
		\draw [bend right=90, looseness=1.50] (10.center) to (11.center);
		\draw (13.center) to (14);
		\draw (16.center) to (17);
	\end{pgfonlayer}
\end{tikzpicture}
\end{equation}
and finally the Kraus decomposition of a map can be expressed using the bent wire.
\begin{equation}
    \begin{tikzpicture}[scale=1, {every node/.style}={scale=0.9}]
	\begin{pgfonlayer}{nodelayer}
		\node [style=morphconj] (49) at (0.25, -0.5) {$K$};
		\node [style=none] (50) at (0.25, -1.25) {};
		\node [style=morph] (51) at (1.75, -0.5) {$K$};
		\node [style=none] (52) at (1.75, -1.25) {};
		\node [style=none] (53) at (0.25, 0.25) {};
		\node [style=none] (54) at (1.75, 0.25) {};
		\node [style=none] (56) at (5, -0.5) {$\sum_{i} K_{i} (-) K^{\dagger}_{i}$};
		\node [style=none] (57) at (-0.25, -1.25) {};
		\node [style=none] (58) at (2.25, -1.25) {};
		\node [style=morphconj] (59) at (-4.75, -0.5) {$K$};
		\node [style=none] (60) at (-4.75, -1.25) {};
		\node [style=morph] (61) at (-3.25, -0.5) {$K$};
		\node [style=none] (62) at (-3.25, -1.25) {};
		\node [style=none] (63) at (-4.75, 0.25) {};
		\node [style=none] (64) at (-3.25, 0.25) {};
		\node [style=none] (65) at (-5.25, -1.25) {};
		\node [style=none] (66) at (-2.75, -1.25) {};
		\node [style=none] (67) at (-2, -0.5) {=};
		\node [style=none] (68) at (-1, -0.5) {$\sum_{i}$};
		\node [style=none] (69) at (0.75, 0.25) {};
		\node [style=none] (70) at (1.25, 0.25) {};
		\node [style=keffectconj] (71) at (0.7, 0.25) {$i$};
		\node [style=keffect] (72) at (1.3, 0.25) {$i$};
		\node [style=none] (73) at (3, -0.5) {$\equiv$};
		\node [style=none] (74) at (-6.25, -0.5) {$\equiv$};
		\node [style=none] (75) at (-7.75, -0.5) {$\pazocal{N}(-)$};
	\end{pgfonlayer}
	\begin{pgfonlayer}{edgelayer}
		\draw (49) to (50.center);
		\draw (52.center) to (51);
		\draw (53.center) to (49);
		\draw (51) to (54.center);
		\draw (59) to (60.center);
		\draw (62.center) to (61);
		\draw (63.center) to (59);
		\draw (61) to (64.center);
		\draw [bend left=60] (59) to (61);
		\draw [in=-105, out=60] (49) to (69.center);
		\draw [in=120, out=-90] (70.center) to (51);
	\end{pgfonlayer}
\end{tikzpicture}
\end{equation}
Formally the bent wire representation of a quantum channel is known as the CPM-construction \cite{Selinger2007DaggerAbstract}. The representation of the trace as a wire gives an intuition for the wherabouts of the information flow in indefinite causal order scenarios.

\section{Translating the Output of the Quantum Switch into a Sum of Diagrams}
We notate each of the two quantum channels in the input of the quantum switch as $\{\pazocal{N}^{(i)}\}_{i=\{1,2\}}$ by

\begin{equation}
    \begin{tikzpicture}[scale=1, {every node/.style}={scale=0.9}]
	\begin{pgfonlayer}{nodelayer}
		\node [style=morphconj] (59) at (0.75, 0) {$i$};
		\node [style=none] (60) at (0.75, -0.75) {};
		\node [style=morph] (61) at (2.25, 0) {$i$};
		\node [style=none] (62) at (2.25, -0.75) {};
		\node [style=none] (63) at (0.75, 0.75) {};
		\node [style=none] (64) at (2.25, 0.75) {};
		\node [style=none] (65) at (0.25, -0.75) {};
		\node [style=none] (66) at (2.75, -0.75) {};
		\node [style=none] (74) at (-0.75, 0) {$\equiv$};
		\node [style=none] (75) at (-2.25, 0) {$\pazocal{N}^{(i)}$};
	\end{pgfonlayer}
	\begin{pgfonlayer}{edgelayer}
		\draw (59) to (60.center);
		\draw (62.center) to (61);
		\draw (63.center) to (59);
		\draw (61) to (64.center);
		\draw [bend left=60] (59) to (61);
	\end{pgfonlayer}
\end{tikzpicture}
\end{equation}

By replacing sums over Kraus operators with caps, the output of the quantum switch of $\pazocal{N}^{(1)}$ and $\pazocal{N}^{(2)}$ can then be written as a sum of diagrams, each of which we refer to as ``CPM-like":
\begin{equation}\label{qmsbasis}
  \begin{tikzpicture}[scale=0.8, {every node/.style}={scale=0.8}]
	\begin{pgfonlayer}{nodelayer}
		\node [style=morphconj] (0) at (-8.25, -0.5) {$1$};
		\node [style=morphconj] (1) at (-8.25, -1.75) {$2$};
		\node [style=morph] (2) at (-6.75, -0.5) {$1$};
		\node [style=morph] (3) at (-6.75, -1.75) {$2$};
		\node [style=none] (4) at (-8.25, 0.25) {};
		\node [style=none] (5) at (-6.75, 0.25) {};
		\node [style=none] (6) at (-8.25, -2.5) {};
		\node [style=none] (7) at (-6.75, -2.5) {};
		\node [style=kstateconj] (8) at (-6, -2.5) {0};
		\node [style=none] (10) at (-6, -1.75) {};
		\node [style=kstate] (11) at (-5.25, -2.5) {0};
		\node [style=none] (12) at (-5.25, -1.75) {};
		\node [style=morphconj] (25) at (2, -0.5) {$2$};
		\node [style=morphconj] (26) at (2, -1.75) {$1$};
		\node [style=morph] (27) at (3.5, -0.5) {$1$};
		\node [style=morph] (28) at (3.5, -1.75) {$2$};
		\node [style=none] (29) at (2, 0.25) {};
		\node [style=none] (30) at (3.5, 0.25) {};
		\node [style=none] (31) at (2, -2.5) {};
		\node [style=none] (32) at (3.5, -2.5) {};
		\node [style=kstateconj] (33) at (4.25, -2.5) {1};
		\node [style=none] (34) at (4.25, -1.75) {};
		\node [style=kstate] (35) at (5, -2.5) {0};
		\node [style=none] (36) at (5, -1.75) {};
		\node [style=morphconj] (37) at (7, -0.5) {$1$};
		\node [style=morphconj] (38) at (7, -1.75) {$2$};
		\node [style=morph] (39) at (8.5, -0.5) {$2$};
		\node [style=morph] (40) at (8.5, -1.75) {$1$};
		\node [style=none] (41) at (7, 0.25) {};
		\node [style=none] (42) at (8.5, 0.25) {};
		\node [style=none] (43) at (7, -2.5) {};
		\node [style=none] (44) at (8.5, -2.5) {};
		\node [style=kstateconj] (45) at (9.25, -2.5) {0};
		\node [style=none] (46) at (9.25, -1.75) {};
		\node [style=kstate] (47) at (10, -2.5) {1};
		\node [style=none] (48) at (10, -1.75) {};
		\node [style=morphconj] (49) at (-3, -0.5) {$2$};
		\node [style=morphconj] (50) at (-3, -1.75) {$1$};
		\node [style=morph] (51) at (-1.5, -0.5) {$2$};
		\node [style=morph] (52) at (-1.5, -1.75) {$1$};
		\node [style=none] (53) at (-3, 0.25) {};
		\node [style=none] (54) at (-1.5, 0.25) {};
		\node [style=none] (55) at (-3, -2.5) {};
		\node [style=none] (56) at (-1.5, -2.5) {};
		\node [style=kstateconj] (57) at (-0.75, -2.5) {1};
		\node [style=none] (58) at (-0.75, -1.75) {};
		\node [style=kstate] (59) at (0, -2.5) {1};
		\node [style=none] (60) at (0, -1.75) {};
		\node [style=none] (61) at (-4.5, -1) {+};
		\node [style=none] (63) at (0.5, -1) {+};
		\node [style=none] (64) at (5.5, -1) {+};
		\node [style=none] (65) at (-5.5, -1) {$\frac{1}{2}$};
		\node [style=none] (69) at (-0.25, -1) {$\frac{1}{2}$};
		\node [style=none] (70) at (4.75, -1) {$\frac{1}{2}$};
		\node [style=none] (71) at (9.75, -1) {$\frac{1}{2}$};
	\end{pgfonlayer}
	\begin{pgfonlayer}{edgelayer}
		\draw (0) to (1);
		\draw (3) to (2);
		\draw (4.center) to (0);
		\draw (2) to (5.center);
		\draw (3) to (7.center);
		\draw (1) to (6.center);
		\draw (8) to (10.center);
		\draw (11) to (12.center);
		\draw [style=Blue, in=120, out=60, looseness=1.25] (25) to (28);
		\draw [style=Red, in=120, out=60, looseness=1.25] (26) to (27);
		\draw (25) to (26);
		\draw (28) to (27);
		\draw (29.center) to (25);
		\draw (27) to (30.center);
		\draw (28) to (32.center);
		\draw (26) to (31.center);
		\draw (33) to (34.center);
		\draw (35) to (36.center);
		\draw [style=Red, in=120, out=60, looseness=1.25] (37) to (40);
		\draw [style=Blue, in=120, out=60, looseness=1.25] (38) to (39);
		\draw (37) to (38);
		\draw (40) to (39);
		\draw (41.center) to (37);
		\draw (39) to (42.center);
		\draw (40) to (44.center);
		\draw (38) to (43.center);
		\draw (45) to (46.center);
		\draw (47) to (48.center);
		\draw [style=Red, bend left=60] (0) to (2);
		\draw [style=Blue, bend left=60] (1) to (3);
		\draw (49) to (50);
		\draw (52) to (51);
		\draw (53.center) to (49);
		\draw (51) to (54.center);
		\draw (52) to (56.center);
		\draw (50) to (55.center);
		\draw (57) to (58.center);
		\draw (59) to (60.center);
		\draw [style=Blue, in=120, out=60] (49) to (51);
		\draw [style=Red, bend left=60] (50) to (52);
	\end{pgfonlayer}
\end{tikzpicture}
\end{equation}
This picture can be used to reproduce Eq.(\ref{eq2})  and the consequent classical capacity activation. Indeed, taking each of $\pazocal{N}^{(1)}$ and $\pazocal{N}^{(2)}$ to be CDPCs, a CDPC is written as 
\begin{equation}
    \begin{tikzpicture}[scale=0.6]
	\begin{pgfonlayer}{nodelayer}
		\node [style=none] (0) at (-0.5, 1.5) {};
		\node [style=none] (1) at (2.75, 1.5) {};
		\node [style=none] (2) at (-0.5, -1.5) {};
		\node [style=none] (3) at (2.75, -1.5) {};
		\node [style=none] (4) at (-1, 0.25) {$\frac{1}{d}$};
		\node [style=none] (5) at (-2, 0.25) {=};
		\node [style=none] (6) at (-3.25, 0.25) {$\pazocal{D}$};
	\end{pgfonlayer}
	\begin{pgfonlayer}{edgelayer}
		\draw [bend left=90, looseness=1.25] (2.center) to (3.center);
		\draw [bend right=90, looseness=1.25] (0.center) to (1.center);
	\end{pgfonlayer}
\end{tikzpicture}
\end{equation}
The diagram for a CDPC separates vertically, which implies that it has no classical or quantum capacity. Upon insertion into equation \ref{qmsbasis},
\begin{equation}
  \begin{tikzpicture}[scale=0.8, {every node/.style}={scale=0.8}]
	\begin{pgfonlayer}{nodelayer}
		\node [style=kstateconj] (8) at (-6.25, -2.5) {0};
		\node [style=none] (10) at (-6.25, -1.75) {};
		\node [style=kstate] (11) at (-5.5, -2.5) {0};
		\node [style=none] (12) at (-5.5, -1.75) {};
		\node [style=kstateconj] (33) at (3.25, -2.5) {1};
		\node [style=none] (34) at (3.25, -1.75) {};
		\node [style=kstate] (35) at (4, -2.5) {0};
		\node [style=none] (36) at (4, -1.75) {};
		\node [style=kstateconj] (45) at (8, -2.5) {0};
		\node [style=none] (46) at (8, -1.75) {};
		\node [style=kstate] (47) at (8.75, -2.5) {1};
		\node [style=none] (48) at (8.75, -1.75) {};
		\node [style=kstateconj] (57) at (-1.5, -2.5) {1};
		\node [style=none] (58) at (-1.5, -1.75) {};
		\node [style=kstate] (59) at (-0.75, -2.5) {1};
		\node [style=none] (60) at (-0.75, -1.75) {};
		\node [style=none] (61) at (-4.5, -1) {+};
		\node [style=none] (63) at (0.25, -1) {+};
		\node [style=none] (64) at (5, -1) {+};
		\node [style=none] (65) at (-5.75, -1) {$\frac{1}{2d^2}$};
		\node [style=none] (66) at (-0.75, -0.75) {};
		\node [style=none] (67) at (4, -0.75) {};
		\node [style=none] (68) at (8.75, -0.75) {};
		\node [style=none] (69) at (-1, -1) {$\frac{1}{2d^2}$};
		\node [style=none] (70) at (3.75, -1) {$\frac{1}{2d^2}$};
		\node [style=none] (71) at (8.5, -1) {$\frac{1}{2d^2}$};
		\node [style=none] (72) at (-8.5, 0) {};
		\node [style=none] (73) at (-7, 0) {};
		\node [style=none] (74) at (-8.5, -1.25) {};
		\node [style=none] (75) at (-7, -1.25) {};
		\node [style=none] (76) at (-8.5, -1.25) {};
		\node [style=none] (77) at (-7, -1.25) {};
		\node [style=none] (78) at (-8.5, -2.5) {};
		\node [style=none] (79) at (-7, -2.5) {};
		\node [style=none] (80) at (-3.75, -2.5) {};
		\node [style=none] (81) at (-2.25, -2.5) {};
		\node [style=none] (82) at (-3.75, -1.25) {};
		\node [style=none] (83) at (-2.25, -1.25) {};
		\node [style=none] (84) at (-3.75, -1.25) {};
		\node [style=none] (85) at (-2.25, -1.25) {};
		\node [style=none] (86) at (-3.75, 0) {};
		\node [style=none] (87) at (-2.25, 0) {};
		\node [style=none] (88) at (1, 0) {};
		\node [style=none] (89) at (2.5, -1.5) {};
		\node [style=none] (90) at (1, -1) {};
		\node [style=none] (91) at (2.5, -2.5) {};
		\node [style=none] (92) at (2.5, 0) {};
		\node [style=none] (93) at (2.5, -1) {};
		\node [style=none] (94) at (1, -1.5) {};
		\node [style=none] (95) at (1, -2.5) {};
		\node [style=none] (96) at (7.25, 0) {};
		\node [style=none] (97) at (5.75, -1.5) {};
		\node [style=none] (98) at (7.25, -1) {};
		\node [style=none] (99) at (5.75, -2.5) {};
		\node [style=none] (100) at (5.75, 0) {};
		\node [style=none] (101) at (5.75, -1) {};
		\node [style=none] (102) at (7.25, -1.5) {};
		\node [style=none] (103) at (7.25, -2.5) {};
	\end{pgfonlayer}
	\begin{pgfonlayer}{edgelayer}
		\draw (8) to (10.center);
		\draw (11) to (12.center);
		\draw (33) to (34.center);
		\draw (35) to (36.center);
		\draw (45) to (46.center);
		\draw (47) to (48.center);
		\draw (57) to (58.center);
		\draw (59) to (60.center);
		\draw [style=Red, bend right=90] (72.center) to (73.center);
		\draw [style=Red, bend left=90] (74.center) to (75.center);
		\draw [style=Blue, bend right=90] (76.center) to (77.center);
		\draw [style=Blue, bend left=90] (78.center) to (79.center);
		\draw [style=Red, bend left=90] (80.center) to (81.center);
		\draw [style=Red, bend right=90] (82.center) to (83.center);
		\draw [style=Blue, bend left=90] (84.center) to (85.center);
		\draw [style=Blue, bend right=90] (86.center) to (87.center);
		\draw [style=Blue, bend right=45, looseness=0.75] (88.center) to (89.center);
		\draw [style=Blue, bend left=45, looseness=0.75] (90.center) to (91.center);
		\draw [style=Red, bend right=45, looseness=0.75] (94.center) to (92.center);
		\draw [style=Red, bend left=45, looseness=0.75] (95.center) to (93.center);
		\draw [bend left=90, looseness=1.50] (94.center) to (90.center);
		\draw [bend right=90, looseness=1.25] (89.center) to (93.center);
		\draw [style=Blue, bend left=45, looseness=0.75] (96.center) to (97.center);
		\draw [style=Blue, bend right=45, looseness=0.75] (98.center) to (99.center);
		\draw [style=Red, bend left=45, looseness=0.75] (102.center) to (100.center);
		\draw [style=Red, bend right=45, looseness=0.75] (103.center) to (101.center);
		\draw [bend right=90, looseness=1.50] (102.center) to (98.center);
		\draw [bend left=90, looseness=1.25] (97.center) to (101.center);
	\end{pgfonlayer}
\end{tikzpicture}
\end{equation}
The above diagrams provide an intuition for capacity activation: the crossing over of the depolarising channels in the last two terms of the sum  allows information to flow from the input (at the bottom of the picture)  to the output (at the top of the picture). 
\section{Alternative Intuition for Capacity Activation}

From the diagrammatic presentation of capacity activation, it is clear that information reaches the output by using environment systems to shuttle between the \textit{bra} and \textit{ket} parts of a density matrix. This intuition may be expanded to give a step-by-step description of the mechanism of capacity enhancement. For each channel $f$ with output $O_{f}$ and environment $E_{f}$, rather than working with the CPM representation, we instead keep track of the environment system $E_{f}$ by working with the Stinespring dilation.  
\begin{equation}
    \begin{tikzpicture}[scale=1, {every node/.style}={scale=0.9}]
	\begin{pgfonlayer}{nodelayer}
		\node [style=morphconj] (59) at (-3, -0.25) {$f$};
		\node [style=none] (60) at (-3, -1) {};
		\node [style=morph] (61) at (-1.5, -0.25) {$f$};
		\node [style=none] (62) at (-1.5, -1) {};
		\node [style=none] (63) at (-3, 0.5) {};
		\node [style=none] (64) at (-1.5, 0.5) {};
		\node [style=none] (65) at (-3, 1) {$O_{f}$};
		\node [style=none] (66) at (-2.25, 0.5) {$E_{f}$};
		\node [style=none] (67) at (-1.5, 1) {$O_{f}$};
		\node [style=morphconj] (68) at (1.5, -0.25) {$f$};
		\node [style=none] (69) at (1.5, -1) {};
		\node [style=none] (70) at (1.5, 0.5) {};
		\node [style=none] (71) at (2.5, 0.5) {};
		\node [style=none] (72) at (1.5, 1) {$O_{f}$};
		\node [style=none] (73) at (2.5, 1) {$E_{f}$};
		\node [style=morphconj] (74) at (4.5, -0.25) {$f$};
		\node [style=none] (75) at (4.5, -1) {};
		\node [style=none] (76) at (4.5, 0.5) {};
		\node [style=none] (77) at (3.5, 0.5) {};
		\node [style=none] (78) at (4.5, 1) {$O_{f}$};
		\node [style=none] (79) at (3.5, 1) {$E_{f}$};
		\node [style=none] (80) at (0, 0) {$\approx$};
	\end{pgfonlayer}
	\begin{pgfonlayer}{edgelayer}
		\draw (59) to (60.center);
		\draw (62.center) to (61);
		\draw (63.center) to (59);
		\draw (61) to (64.center);
		\draw [bend left=60] (59) to (61);
		\draw (68) to (70.center);
		\draw (69.center) to (68);
		\draw [in=-90, out=60] (68) to (71.center);
		\draw (74) to (76.center);
		\draw (75.center) to (74);
		\draw [in=-90, out=120] (74) to (77.center);
	\end{pgfonlayer}
\end{tikzpicture}
\end{equation}
The depolarising channel can be dilated to an isometry that sends the input state into an environment and entangles the output with an independent environment,
\begin{equation}
    \begin{tikzpicture}
	\begin{pgfonlayer}{nodelayer}
		\node [style=morphconj] (0) at (-4, 0.25) {$f$};
		\node [style=none] (2) at (-4, -0.5) {};
		\node [style=none] (4) at (-4, 1.25) {};
		\node [style=none] (5) at (-3, 1.25) {};
		\node [style=none] (6) at (-4, 1.75) {$O_{f}$};
		\node [style=none] (7) at (-3, 1.75) {$E_{f}$};
		\node [style=morphconj] (8) at (-1, 0.25) {$f$};
		\node [style=none] (9) at (-1, -0.5) {};
		\node [style=none] (10) at (-1, 1.25) {};
		\node [style=none] (11) at (-2, 1.25) {};
		\node [style=none] (12) at (-1, 1.75) {$O_{f}$};
		\node [style=none] (13) at (-2, 1.75) {$E_{f}$};
		\node [style=none] (14) at (0, 0.5) {$=$};
		\node [style=none] (15) at (0.75, -0.5) {};
		\node [style=none] (16) at (0.75, 1.25) {};
		\node [style=none] (17) at (1.75, 1.25) {};
		\node [style=none] (18) at (2.5, 1.25) {};
		\node [style=none] (19) at (4.75, -0.5) {};
		\node [style=none] (20) at (4.75, 1.25) {};
		\node [style=none] (21) at (4, 1.25) {};
		\node [style=none] (22) at (3.25, 1.25) {};
		\node [style=none] (23) at (0.5, 0.5) {$\frac{1}{\sqrt{d}}$};
		\node [style=none] (24) at (5.25, 0.5) {$\frac{1}{\sqrt{d}}$};
		\node [style=none] (25) at (1.75, 1.75) {$E_{f1}$};
		\node [style=none] (26) at (2.5, 1.75) {$E_{f2}$};
		\node [style=none] (27) at (0.75, 1.75) {$O$};
		\node [style=none] (28) at (4, 1.75) {$E_{f1}$};
		\node [style=none] (29) at (3.25, 1.75) {$E_{f2}$};
		\node [style=none] (30) at (4.75, 1.75) {$O$};
	\end{pgfonlayer}
	\begin{pgfonlayer}{edgelayer}
		\draw (0) to (4.center);
		\draw (2.center) to (0);
		\draw [in=-90, out=60] (0) to (5.center);
		\draw (8) to (10.center);
		\draw (9.center) to (8);
		\draw [in=-90, out=120] (8) to (11.center);
		\draw [in=90, out=-90, looseness=0.75] (17.center) to (15.center);
		\draw [in=-90, out=-90, looseness=2.00] (16.center) to (18.center);
		\draw [in=90, out=-90, looseness=0.75] (21.center) to (19.center);
		\draw [in=-90, out=-90, looseness=2.00] (20.center) to (22.center);
	\end{pgfonlayer}
\end{tikzpicture}
\end{equation}
where the upwards pointing cup represents a Bell state. In the following we will use this diagram to visualise the information flow in the switch of completely depolarising channels. 

\subsection{The Switch of Depolarising Channels as Superposition of the Whereabouts of the Input}

Denoting the environments of $f$ and $g$ as $E_{f}$ and $E_{g}$, and their outputs as $O_{f}$,$O_{g}$, we consider the interference term of the quantum switch of the dilations of $f$ and $g$
\begin{equation}
    \begin{tikzpicture}
	\begin{pgfonlayer}{nodelayer}
		\node [style=morphconj] (0) at (-3, 0.25) {$f$};
		\node [style=none] (2) at (-3, -0.75) {};
		\node [style=none] (4) at (-3, 1.25) {};
		\node [style=none] (5) at (-0.5, 1.25) {};
		\node [style=morphconj] (6) at (-3, 1.25) {$g$};
		\node [style=none] (7) at (-3, 1.25) {};
		\node [style=none] (8) at (-3, 2.25) {};
		\node [style=none] (9) at (-1.75, 2.25) {};
		\node [style=none] (10) at (-0.5, 2.25) {};
		\node [style=none] (11) at (4.75, 0.25) {};
		\node [style=kstateconj] (12) at (4.75, -0.75) {$0$};
		\node [style=morphtrans] (13) at (3.75, 0.25) {$g$};
		\node [style=none] (14) at (3.75, -0.75) {};
		\node [style=none] (15) at (3.75, 1.25) {};
		\node [style=none] (16) at (1.25, 1.25) {};
		\node [style=morphtrans] (17) at (3.75, 1.25) {$f$};
		\node [style=none] (18) at (3.75, 1.25) {};
		\node [style=none] (19) at (3.75, 2.25) {};
		\node [style=none] (20) at (1.25, 2.25) {};
		\node [style=none] (21) at (2.5, 2.25) {};
		\node [style=none] (22) at (5.5, 0.25) {};
		\node [style=kstate] (23) at (5.5, -0.75) {$1$};
		\node [style=none] (25) at (5, 1.25) {$\frac{1}{2}$};
		\node [style=none] (28) at (-3, 2.75) {$O_{g}$};
		\node [style=none] (29) at (-1.75, 2.75) {$E_{g}$};
		\node [style=none] (30) at (-0.5, 2.75) {$E_{f}$};
		\node [style=none] (34) at (3.75, 2.75) {$O_{f}$};
		\node [style=none] (35) at (2.5, 2.75) {$E_{g}$};
		\node [style=none] (36) at (1.25, 2.75) {$E_{f}$};
	\end{pgfonlayer}
	\begin{pgfonlayer}{edgelayer}
		\draw (0) to (4.center);
		\draw (2.center) to (0);
		\draw [in=-90, out=60, looseness=0.75] (0) to (5.center);
		\draw (6) to (8.center);
		\draw (7.center) to (6);
		\draw [in=-90, out=60, looseness=0.75] (6) to (9.center);
		\draw (5.center) to (10.center);
		\draw (12) to (11.center);
		\draw (13) to (15.center);
		\draw (14.center) to (13);
		\draw [in=-90, out=120, looseness=0.75] (13) to (16.center);
		\draw (17) to (19.center);
		\draw (18.center) to (17);
		\draw [in=-90, out=120, looseness=0.75] (17) to (20.center);
		\draw [in=-90, out=90] (16.center) to (21.center);
		\draw (23) to (22.center);
	\end{pgfonlayer}
\end{tikzpicture}
\end{equation}



For $f$ and $g$ dilations of completely depolarising channels this gives
\begin{equation}\label{eq15}
    \begin{tikzpicture}
	\begin{pgfonlayer}{nodelayer}
		\node [style=none] (10) at (0, -0.5) {};
		\node [style=none] (11) at (6.5, -0.5) {};
		\node [style=kstateconj] (12) at (6.5, -1.5) {$0$};
		\node [style=none] (22) at (7.5, -0.5) {};
		\node [style=kstate] (23) at (7.5, -1.5) {$1$};
		\node [style=none] (25) at (6.75, 0.5) {$\frac{1}{2}$};
		\node [style=none] (28) at (-4, 0) {$O_{g}$};
		\node [style=none] (29) at (-2.75, 0) {$E_{g1}$};
		\node [style=none] (34) at (5.25, 0) {$O_{f}$};
		\node [style=none] (36) at (2, 0) {$E_{f1}$};
		\node [style=none] (38) at (-4, -0.5) {};
		\node [style=none] (39) at (-2.75, -0.5) {};
		\node [style=none] (40) at (-2, -0.5) {};
		\node [style=kstateconj] (41) at (-0.75, -1.25) {$\psi$};
		\node [style=none] (45) at (-0.75, -0.5) {};
		\node [style=kstate] (53) at (4, -1.25) {$\psi$};
		\node [style=none] (57) at (5.25, -0.5) {};
		\node [style=none] (58) at (4, -0.5) {};
		\node [style=none] (59) at (3.25, -0.5) {};
		\node [style=none] (60) at (2, -0.5) {};
		\node [style=none] (61) at (1.25, -0.5) {};
		\node [style=none] (62) at (-2, 0) {$E_{g2}$};
		\node [style=none] (63) at (4, 0) {$E_{g1}$};
		\node [style=none] (64) at (3.25, 0) {$E_{g2}$};
		\node [style=none] (65) at (-0.75, 0) {$E_{f1}$};
		\node [style=none] (66) at (0, 0) {$E_{f2}$};
		\node [style=none] (67) at (1.25, 0) {$E_{f2}$};
	\end{pgfonlayer}
	\begin{pgfonlayer}{edgelayer}
		\draw (12) to (11.center);
		\draw (23) to (22.center);
		\draw [in=-90, out=-90, looseness=2.00] (38.center) to (40.center);
		\draw (41) to (45.center);
		\draw [bend right=90, looseness=2.00] (39.center) to (10.center);
		\draw [bend right=90] (61.center) to (57.center);
		\draw [bend right=90, looseness=2.00] (60.center) to (59.center);
		\draw (53) to (58.center);
	\end{pgfonlayer}
\end{tikzpicture}
\end{equation}
Now, when a state is inserted into the input of the channel, its whereabouts is determined by  the control qubit: if the control state is $\ket{0}$, then the state $\ket{\psi}$ is in $E_{f1}$; if the control state is $\ket{1}$, then the state $\ket{\psi}$ is in $E_{g1}$.

If we now trace all the environment systems, we obtain the following diagram
\begin{equation}\label{teleportation}
    \begin{tikzpicture}
	\begin{pgfonlayer}{nodelayer}
		\node [style=none] (10) at (0, -0.5) {};
		\node [style=none] (11) at (6.5, -0.5) {};
		\node [style=kstateconj] (12) at (6.5, -1.5) {$0$};
		\node [style=none] (22) at (7.5, -0.5) {};
		\node [style=kstate] (23) at (7.5, -1.5) {$1$};
		\node [style=none] (25) at (6.75, 0.5) {$\frac{1}{2}$};
		\node [style=none] (28) at (-4, 0) {$O_{g}$};
		\node [style=none] (34) at (5.25, 0) {$O_{f}$};
		\node [style=none] (38) at (-4, -0.5) {};
		\node [style=none] (39) at (-2.75, -0.5) {};
		\node [style=none] (40) at (-2, -0.5) {};
		\node [style=kstateconj] (41) at (-0.75, -1.25) {$\psi$};
		\node [style=none] (45) at (-0.75, -0.5) {};
		\node [style=kstate] (53) at (4, -1.25) {$\psi$};
		\node [style=none] (57) at (5.25, -0.5) {};
		\node [style=none] (58) at (4, -0.5) {};
		\node [style=none] (59) at (3.25, -0.5) {};
		\node [style=none] (60) at (2, -0.5) {};
		\node [style=none] (61) at (1.25, -0.5) {};
	\end{pgfonlayer}
	\begin{pgfonlayer}{edgelayer}
		\draw (12) to (11.center);
		\draw (23) to (22.center);
		\draw [in=-90, out=-90, looseness=2.00] (38.center) to (40.center);
		\draw (41) to (45.center);
		\draw [bend right=90, looseness=2.00] (39.center) to (10.center);
		\draw [bend right=90] (61.center) to (57.center);
		\draw [bend right=90, looseness=2.00] (60.center) to (59.center);
		\draw (53) to (58.center);
		\draw [bend left=90, looseness=1.25] (45.center) to (60.center);
		\draw [bend right=90] (59.center) to (40.center);
		\draw [bend left=90, looseness=1.50] (10.center) to (61.center);
		\draw [bend left=90] (39.center) to (58.center);
	\end{pgfonlayer}
\end{tikzpicture}
\end{equation}
which shows the flow of information from the input to the output arising from the interference of the diagrams corresponding to states $|0\rangle$ and $|1\rangle$.   The information flow is implemented by teleportation, which makes the state $|\psi\rangle$ ricochet on the environment and reappear on the output.  
 


\section{N-Channel Switches}
A natural generalisation of the quantum switch is a quantum supermap that adds quantum  control to some choice of sequential orders of $3$ or more channels \cite{colnaghi2012quantum, Facchini2015QuantumProblem}. We could imagine for example using states $\ket{0},\ket{1},\ket{2}$ of a control qutrit to implement the sequential compositions ($\pazocal{N}^{(1)} \circ \pazocal{N}^{(2)} \circ \pazocal{N}^{(3)} $), ($\pazocal{N}^{(3)} \circ \pazocal{N}^{(1)} \circ \pazocal{N}^{(2)} $), and ($\pazocal{N}^{(2)} \circ \pazocal{N}^{(3)} \circ \pazocal{N}^{(1)} $) respectively. Again one would expect to see interference between the two choices of sequential order in the left and right hand sides of an interference term of the control, 
\begin{figure}[H]
    \centering
    \begin{equation}
    \begin{tikzpicture}
	\begin{pgfonlayer}{nodelayer}
		\node [style=morphconj] (12) at (0.5, 0.25) {$3$};
		\node [style=morphconj] (13) at (0.5, -1) {$1$};
		\node [style=morph] (14) at (1.75, 1.5) {$3$};
		\node [style=morph] (15) at (1.75, -1) {$2$};
		\node [style=none] (18) at (0.5, -1.75) {};
		\node [style=none] (19) at (1.75, -1.75) {};
		\node [style=kstateconj] (20) at (2.75, -1.75) {2};
		\node [style=none] (21) at (2.75, -1) {};
		\node [style=kstate] (22) at (3.3, -1.75) {1};
		\node [style=none] (23) at (3.3, -1) {};
		\node [style=morphconj] (54) at (0.5, 1.5) {$2$};
		\node [style=morph] (57) at (1.75, 0.25) {$1$};
		\node [style=none] (59) at (0.5, 2.25) {};
		\node [style=none] (60) at (1.75, 2.25) {};
	\end{pgfonlayer}
	\begin{pgfonlayer}{edgelayer}
		\draw (12) to (13);
		\draw (15) to (19.center);
		\draw (13) to (18.center);
		\draw (20) to (21.center);
		\draw (22) to (23.center);
		\draw (59.center) to (54);
		\draw (54) to (12);
		\draw (57) to (15);
		\draw (57) to (14);
		\draw (14) to (60.center);
		\draw [style=Red, in=120, out=60] (54) to (15);
		\draw [style=Pink, in=120, out=60, looseness=1.25] (12) to (14);
		\draw [style=Blue, in=120, out=60, looseness=1.25] (13) to (57);
	\end{pgfonlayer}
\end{tikzpicture}
    \end{equation}
    \caption{CPM-like diagram for the term $\ket{2}\bra{1}$ in the control of a switch which in state $\ket{2}$ implements $\pazocal{N}^{(2)} \circ \pazocal{N}^{(3)} \circ \pazocal{N}^{(1)} $, and in state $\ket{1}$ implements $\pazocal{N}^{(3)} \circ \pazocal{N}^{(1)} \circ \pazocal{N}^{(2)} $. Each colored wire represents a sum over Kraus operators.}
    \label{fig:intriswitch}
\end{figure}
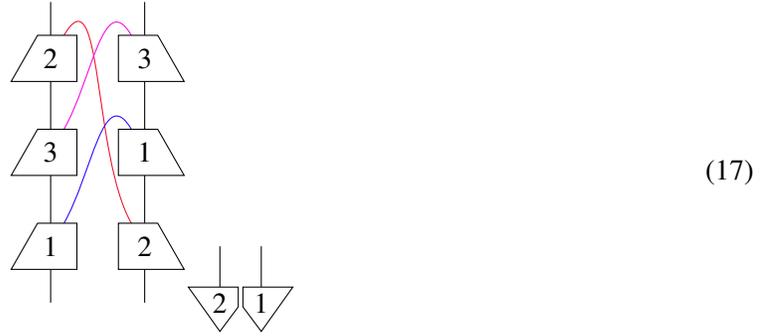
We refer to such a generalisation as a $3$-channel switch. The advantages for classical capacity enhancement of $3$-channel switches were explored in \cite{Procopio2020SendingOrders} where it was observed that the maximum possible Holevo information \cite{Holevo1998TheStatesb} achievable with a superposition of $3$ sequential orders of $3$ CDPCs exceeds the maximal Holevo information achievable with $4$ or $5$ sequential orders of $3$ CDPCs. The Holevo information for $3$ orders of $3$ CDPS's turns out to be maximised when the orders chosen are cyclic permutations  \cite{Procopio2020SendingOrders}. In the following, we will provide a heuristics suggesting that the optimality of cyclic permutation may also hold for  $N\ge 3$ channels.
\subsection{Diagrammatic Presentation of an N-Party Switch}
We call a supermap which generalises the quantum switch to coherent control of $M$ sequential compositions of $N$ channels $\{ \pazocal{N}^{(i)} \}_{i = 1}^{N}$ each with Kraus operators $\{ K_{j_i}^{(i)} \}_{j_i = 1}^{d^2}$ an $N$-channel switch. Labelling $M$ basis states $\{\ket{\pi}\}$ of a quMit control system by the permutations imposed on the sequential composition of the channels $\pazocal{N}^{(\pi(N))} \circ \dots \circ \pazocal{N}^{(\pi(1))}$, then for a Fourier state control $\ket{+}_{M} \equiv \frac{1}{M}\sum_{\pi} \ket{\pi}$ the supermap would implement an equal weighted superposition of each of the $M$ sequential compositions.
\begin{equation}
\begin{split}
\rho' & = S_{M}(\{ \pazocal{N}^{(i)} \}_{i = 1}^{M})(\rho,\ket{+}\bra{+}) \\
& = \frac{1}{M}\sum_{\pi \pi'}  \sum_{j_1}^{d^2}\dots \sum_{j_N}^{d^2} K_{j_{\pi(1)}}^{(\pi(1))} \dots K_{j_{\pi(N)}}^{{(\pi(N))}}\rho K_{j_{\pi'(N)}}^{{(\pi'(N))}^{\dagger}} \dots K_{j_{\pi'(1)}}^{{(\pi'(1))}^{\dagger}} \otimes \ket{\pi}\bra{\pi'} \\
& \equiv  \frac{1}{M}\sum_{\pi \pi'}  \pazocal{N}_{\pi \pi'} \otimes \ket{\pi}\bra{\pi'} 
\end{split}
\end{equation}
As shown in figure \ref{fig:interm}, diagrammatically the term $ \pazocal{N}_{\pi \pi'} \otimes  \ket{\pi}\bra{\pi'} $ of the switch is a CPM-like diagram with boxes rearranged according to permutations $\pi$ and $\pi^{'}$ on the left and right hand wires respectively. Each wire in a CPM-like diagram represents a sum over Kraus operators.
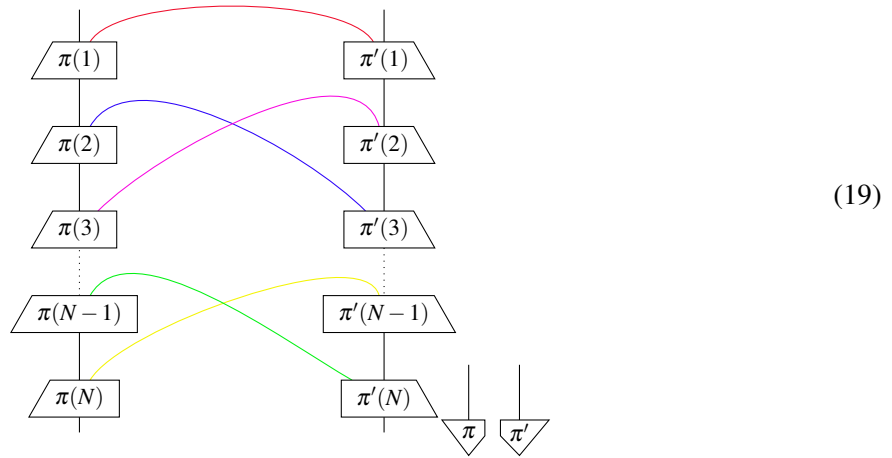
\begin{figure}[H]
    \centering
    \begin{equation}
    \begin{tikzpicture}[scale=0.9, {every node/.style}={scale=0.8}]
	\begin{pgfonlayer}{nodelayer}
		\node [style=morphconj] (179) at (-1.5, 2) {$\pi(1)$};
		\node [style=morph] (180) at (3, 0.75) {$\pi'(2)$};
		\node [style=morph] (181) at (3, 2) {$\pi'(1)$};
		\node [style=morphconj] (182) at (-1.5, 0.75) {$\pi(2)$};
		\node [style=morph] (183) at (3, -0.5) {$\pi'(3)$};
		\node [style=morphconj] (184) at (-1.5, -0.5) {$\pi(3)$};
		\node [style=morph] (185) at (3, -1.75) {$\pi'(N-1)$};
		\node [style=morphconj] (186) at (-1.5, -1.75) {$\pi(N-1)$};
		\node [style=morph] (187) at (3, -3) {$\pi'(N)$};
		\node [style=morphconj] (188) at (-1.5, -3) {$\pi(N)$};
		\node [style=none] (189) at (-1.5, 2.75) {};
		\node [style=none] (190) at (3, 2.75) {};
		\node [style=none] (191) at (-1.5, -3.5) {};
		\node [style=none] (192) at (3, -3.5) {};
		\node [style=kstateconj] (193) at (4.25, -3.5) {$\pi$};
		\node [style=kstate] (194) at (5, -3.5) {$\pi'$};
		\node [style=none] (195) at (4.25, -2.5) {};
		\node [style=none] (196) at (5, -2.5) {};
	\end{pgfonlayer}
	\begin{pgfonlayer}{edgelayer}
		\draw (183) to (180);
		\draw (179) to (182);
		\draw (181) to (180);
		\draw (182) to (184);
		\draw [style=dottededge] (186) to (184);
		\draw [style=Red, bend left=60, looseness=0.50] (179) to (181);
		\draw [style=dottededge, in=90, out=-90] (183) to (185);
		\draw (185) to (187);
		\draw (186) to (188);
		\draw [style=Yellow, in=105, out=60, looseness=0.50] (188) to (185);
		\draw [style=Green, in=150, out=60, looseness=0.75] (186) to (187);
		\draw [style=Blue, in=135, out=60, looseness=0.75] (182) to (183);
		\draw [style=Pink, in=105, out=45, looseness=0.75] (184) to (180);
		\draw (189.center) to (179);
		\draw (190.center) to (181);
		\draw (188) to (191.center);
		\draw (192.center) to (187);
		\draw (193) to (195.center);
		\draw (196.center) to (194);
	\end{pgfonlayer}
\end{tikzpicture}
\end{equation}
    \caption{CPM-like diagram for the term $ \pazocal{N}_{\pi \pi'} \otimes \ket{\pi}\bra{\pi'}$. Each cap algebraically is a sum over Kraus operators for a particular channel, as such the above diagram corresponds to the case in which $\pi(1) = \pi'(1)$,  $\pi(2) = \pi'(3)$, $\pi(3) = \pi'(2)$, $\pi(N-1) = \pi'(N)$, $\pi(N) = \pi'(N-1)$}
    \label{fig:interm}
\end{figure}
\subsection{Capacity Enhancement By Superposition of Cyclic Permutations}
For any term $\pazocal{N}_{\pi \pi'}$ with $\pi$ and $\pi'$ cyclic permutations $\pi$,$\pi'$ are mutually cyclic in the sense that $\pi' \circ \pi^{-1}$ is a cyclic permutation. Taking each channel $\pazocal{N}^{(i)}$ to be a completely depolarising channel, it is quick to see by hand in the $3$-channel case that any such term $\pazocal{N}_{\pi \pi'}$ is proportional to an identity channel. For example, for $\pazocal{N}_{21}$
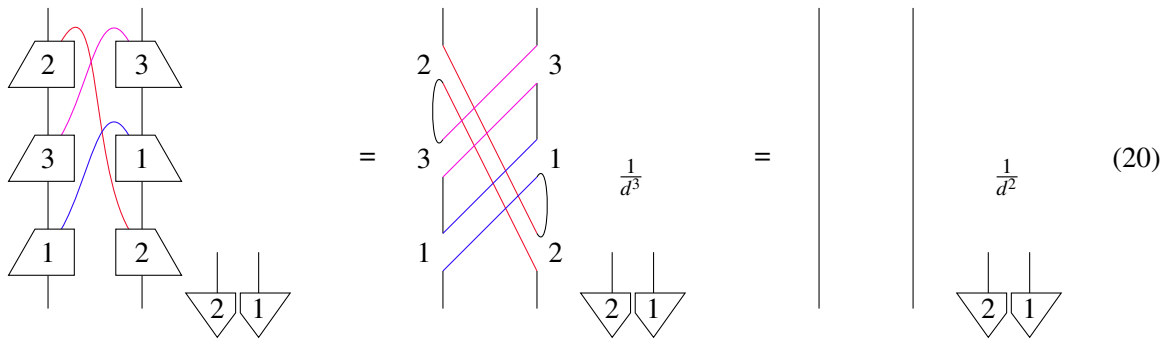
\begin{equation}
    \begin{tikzpicture}
	\begin{pgfonlayer}{nodelayer}
		\node [style=morphconj] (12) at (-6.25, 0) {$3$};
		\node [style=morphconj] (13) at (-6.25, -1.25) {$1$};
		\node [style=morph] (14) at (-5, 1.25) {$3$};
		\node [style=morph] (15) at (-5, -1.25) {$2$};
		\node [style=none] (18) at (-6.25, -2) {};
		\node [style=none] (19) at (-5, -2) {};
		\node [style=kstateconj] (20) at (-4, -2) {2};
		\node [style=none] (21) at (-4, -1.25) {};
		\node [style=kstate] (22) at (-3.45, -2) {1};
		\node [style=none] (23) at (-3.45, -1.25) {};
		\node [style=morphconj] (54) at (-6.25, 1.25) {$2$};
		\node [style=morph] (57) at (-5, 0) {$1$};
		\node [style=none] (59) at (-6.25, 2) {};
		\node [style=none] (60) at (-5, 2) {};
		\node [style=none] (61) at (-2, 0) {=};
		\node [style=none] (66) at (-1, -2) {};
		\node [style=none] (67) at (0.25, -2) {};
		\node [style=kstateconj] (68) at (1.25, -2) {2};
		\node [style=none] (69) at (1.25, -1.25) {};
		\node [style=kstate] (70) at (1.8, -2) {1};
		\node [style=none] (71) at (1.8, -1.25) {};
		\node [style=none] (74) at (-1, 2) {};
		\node [style=none] (75) at (0.25, 2) {};
		\node [style=none] (76) at (3.25, 0) {=};
		\node [style=none] (81) at (4, -2) {};
		\node [style=none] (82) at (5.25, -2) {};
		\node [style=kstateconj] (83) at (6.25, -2) {2};
		\node [style=none] (84) at (6.25, -1.25) {};
		\node [style=kstate] (85) at (6.8, -2) {1};
		\node [style=none] (86) at (6.8, -1.25) {};
		\node [style=none] (89) at (4, 2) {};
		\node [style=none] (90) at (5.25, 2) {};
		\node [style=none] (91) at (-1, 1.5) {};
		\node [style=none] (92) at (-1, 1) {};
		\node [style=none] (93) at (0.25, 1.5) {};
		\node [style=none] (94) at (0.25, 1) {};
		\node [style=none] (95) at (-1, 0.25) {};
		\node [style=none] (96) at (-1, -0.25) {};
		\node [style=none] (97) at (0.25, 0.25) {};
		\node [style=none] (98) at (0.25, -0.25) {};
		\node [style=none] (99) at (0.25, -1) {};
		\node [style=none] (100) at (0.25, -1.5) {};
		\node [style=none] (101) at (-1, -1) {};
		\node [style=none] (102) at (-1, -1.5) {};
		\node [style=none] (103) at (1.5, -0.25) {$\frac{1}{d^3}$};
		\node [style=none] (104) at (6.5, -0.25) {$\frac{1}{d^2}$};
		\node [style=none] (105) at (-1.25, 1.25) {$2$};
		\node [style=none] (106) at (0.5, -1.25) {$2$};
		\node [style=none] (107) at (0.5, 0) {$1$};
		\node [style=none] (108) at (-1.25, -1.25) {$1$};
		\node [style=none] (109) at (-1.25, 0) {$3$};
		\node [style=none] (110) at (0.5, 1.25) {$3$};
	\end{pgfonlayer}
	\begin{pgfonlayer}{edgelayer}
		\draw (12) to (13);
		\draw (15) to (19.center);
		\draw (13) to (18.center);
		\draw (20) to (21.center);
		\draw (22) to (23.center);
		\draw (59.center) to (54);
		\draw (54) to (12);
		\draw (57) to (15);
		\draw (57) to (14);
		\draw (14) to (60.center);
		\draw [style=Red, in=120, out=60] (54) to (15);
		\draw [style=Pink, in=120, out=60, looseness=1.25] (12) to (14);
		\draw [style=Blue, in=120, out=60, looseness=1.25] (13) to (57);
		\draw (68) to (69.center);
		\draw (70) to (71.center);
		\draw (83) to (84.center);
		\draw (85) to (86.center);
		\draw (74.center) to (91.center);
		\draw (75.center) to (93.center);
		\draw [style=Red] (91.center) to (99.center);
		\draw [style=Red] (92.center) to (100.center);
		\draw [bend left=210, looseness=1.25] (92.center) to (95.center);
		\draw (96.center) to (101.center);
		\draw (102.center) to (66.center);
		\draw (100.center) to (67.center);
		\draw [bend right=150, looseness=1.25] (99.center) to (98.center);
		\draw (97.center) to (94.center);
		\draw [style=Pink] (93.center) to (95.center);
		\draw [style=Pink] (94.center) to (96.center);
		\draw [style=Blue] (97.center) to (101.center);
		\draw [style=Blue] (98.center) to (102.center);
		\draw (89.center) to (81.center);
		\draw (82.center) to (90.center);
	\end{pgfonlayer}
\end{tikzpicture}
\end{equation}
The normalisation of this information transmitting term is increased by the presence of the closed loop (the parallelogram in the right-hand-side of the first equality),   which contributes a factor of $d$ to the diagram.
This result immediately generalises, to demonstrate this we first adopt a cleaner notation to cope with the increasing number of boxes
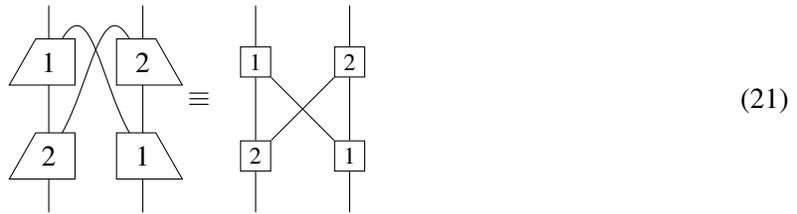
\begin{equation}
    \begin{tikzpicture}
	\begin{pgfonlayer}{nodelayer}
		\node [style=morphconj] (0) at (-2, 0.5) {$1$};
		\node [style=morphconj] (1) at (-2, -0.75) {$2$};
		\node [style=morph] (2) at (-0.75, 0.5) {$2$};
		\node [style=morph] (3) at (-0.75, -0.75) {$1$};
		\node [style=none] (4) at (-2, 1.25) {};
		\node [style=none] (5) at (-0.75, 1.25) {};
		\node [style=none] (6) at (-2, -1.5) {};
		\node [style=none] (7) at (-0.75, -1.5) {};
		\node [style=classhad] (8) at (0.75, 0.5) {$1$};
		\node [style=classhad] (9) at (0.75, -0.75) {$2$};
		\node [style=classhad] (10) at (2, 0.5) {$2$};
		\node [style=classhad] (11) at (2, -0.75) {$1$};
		\node [style=none] (12) at (0.75, 1.25) {};
		\node [style=none] (13) at (2, 1.25) {};
		\node [style=none] (14) at (0.75, -1.5) {};
		\node [style=none] (15) at (2, -1.5) {};
		\node [style=none] (16) at (0, 0) {$\equiv$};
	\end{pgfonlayer}
	\begin{pgfonlayer}{edgelayer}
		\draw [in=120, out=60, looseness=1.25] (0) to (3);
		\draw [in=120, out=60, looseness=1.25] (1) to (2);
		\draw (0) to (1);
		\draw (3) to (2);
		\draw (4.center) to (0);
		\draw (2) to (5.center);
		\draw (3) to (7.center);
		\draw (1) to (6.center);
		\draw (8) to (9);
		\draw (11) to (10);
		\draw (12.center) to (8);
		\draw (10) to (13.center);
		\draw (11) to (15.center);
		\draw (9) to (14.center);
		\draw (11) to (8);
		\draw (9) to (10);
	\end{pgfonlayer}
\end{tikzpicture}
\end{equation}
Then figure \ref{fig:cpi} presents a generic diagram for a cyclic permutation between left and right hand wires.
\begin{figure}[H]
    \centering
    \begin{equation}
        \begin{tikzpicture}[scale=0.8, {every node/.style}={scale=0.8}]
	\begin{pgfonlayer}{nodelayer}
		\node [style=classhad] (13) at (-9.5, -2.5) {$N$};
		\node [style=classhad] (15) at (-6.5, -0.25) {$1$};
		\node [style=none] (18) at (-9.5, -3.75) {};
		\node [style=kstateconj] (20) at (-5.75, -1.25) {$\pi$};
		\node [style=none] (21) at (-5.75, -0.5) {};
		\node [style=kstate] (22) at (-4.95, -1.25) {$\pi'$};
		\node [style=none] (23) at (-4.95, -0.5) {};
		\node [style=classhad] (54) at (-9.5, 3) {$1$};
		\node [style=classhad] (57) at (-6.5, 0.75) {$N$};
		\node [style=none] (59) at (-9.5, 3.75) {};
		\node [style=none] (62) at (-4.25, 0.25) {=};
		\node [style=none] (65) at (-2.5, -3.75) {};
		\node [style=kstateconj] (66) at (1.5, -1.25) {$\pi$};
		\node [style=none] (67) at (1.5, -0.5) {};
		\node [style=kstate] (68) at (2.3, -1.25) {$\pi'$};
		\node [style=none] (69) at (2.3, -0.5) {};
		\node [style=none] (72) at (-2.5, 3.75) {};
		\node [style=none] (73) at (0.5, 3.75) {};
		\node [style=none] (74) at (0.5, -3.75) {};
		\node [style=none] (75) at (-2.5, 3.25) {};
		\node [style=none] (76) at (-2.5, 2.75) {};
		\node [style=none] (77) at (-2.5, 2.25) {};
		\node [style=none] (78) at (-2.5, -1.75) {};
		\node [style=none] (79) at (-2.5, -2.25) {};
		\node [style=none] (80) at (-2.5, -2.75) {};
		\node [style=none] (81) at (0.5, 0.5) {};
		\node [style=none] (82) at (0.5, 1) {};
		\node [style=none] (83) at (0.5, -0.5) {};
		\node [style=none] (85) at (0.5, -1) {};
		\node [style=none] (87) at (-2.75, 3) {$1$};
		\node [style=none] (88) at (0.75, 0.75) {$N$};
		\node [style=none] (91) at (0.75, -0.75) {$1$};
		\node [style=none] (92) at (-3, -2.5) {$N$};
		\node [style=none] (95) at (4.5, -2) {};
		\node [style=kstateconj] (96) at (6.5, -1.25) {$\pi$};
		\node [style=none] (97) at (6.5, -0.5) {};
		\node [style=kstate] (98) at (7.3, -1.25) {$\pi'$};
		\node [style=none] (99) at (7.3, -0.5) {};
		\node [style=none] (102) at (4.5, 2.5) {};
		\node [style=none] (103) at (5.75, 2.5) {};
		\node [style=none] (104) at (5.75, -2) {};
		\node [style=none] (105) at (3.5, 0.25) {$=$};
		\node [style=classhad] (106) at (-6.5, -1.25) {$2$};
		\node [style=classhad] (107) at (-9.5, 2) {$2$};
		\node [style=classhad] (108) at (-6.5, -2.25) {$3$};
		\node [style=classhad] (109) at (-9.5, 1) {$3$};
		\node [style=classhad] (110) at (-9.5, -1.5) {$N-1$};
		\node [style=classhad] (112) at (-6.5, 1.75) {$N-1$};
		\node [style=classhad] (113) at (-9.5, -0.5) {$N-2$};
		\node [style=classhad] (115) at (-6.5, 2.75) {$N-2$};
		\node [style=none] (116) at (-6.5, 3.75) {};
		\node [style=none] (117) at (-6.5, -3.75) {};
		\node [style=none] (118) at (0.5, -1.5) {};
		\node [style=none] (120) at (0.5, -2) {};
		\node [style=none] (121) at (0.5, -2.5) {};
		\node [style=none] (123) at (0.5, 1.5) {};
		\node [style=none] (124) at (-2.5, -0.25) {};
		\node [style=none] (125) at (-2.5, 0) {};
		\node [style=none] (128) at (-2.5, 1.75) {};
		\node [style=none] (129) at (-2.5, 1.25) {};
		\node [style=none] (130) at (-2.75, 2) {$2$};
		\node [style=none] (131) at (0.75, -1.75) {$2$};
		\node [style=none] (132) at (-2.5, 0.75) {};
		\node [style=none] (133) at (-2.5, 0.5) {};
		\node [style=none] (134) at (0.5, -3) {};
		\node [style=none] (135) at (0.5, -3.25) {};
		\node [style=none] (136) at (0.75, -2.75) {$3$};
		\node [style=none] (137) at (-2.75, 1) {$3$};
		\node [style=none] (138) at (0.5, 2) {};
		\node [style=none] (139) at (0.5, 2.5) {};
		\node [style=none] (140) at (-2.5, -1.25) {};
		\node [style=none] (141) at (-2.5, -0.75) {};
		\node [style=none] (142) at (1, 1.75) {$N-1$};
		\node [style=none] (143) at (-3, -1.75) {$N-1$};
		\node [style=none] (144) at (0.5, 3) {};
		\node [style=none] (145) at (0.5, 3.5) {};
		\node [style=none] (146) at (1, 2.75) {$N-2$};
		\node [style=none] (147) at (-3, -0.5) {$N-2$};
		\node [style=none] (148) at (2, 0.5) {$\frac{1}{d^{N}}$};
		\node [style=none] (149) at (6.75, 0.5) {$\frac{1}{d^2}$};
	\end{pgfonlayer}
	\begin{pgfonlayer}{edgelayer}
		\draw (13) to (18.center);
		\draw (20) to (21.center);
		\draw (22) to (23.center);
		\draw (59.center) to (54);
		\draw (57) to (15);
		\draw (66) to (67.center);
		\draw (68) to (69.center);
		\draw (77.center) to (76.center);
		\draw (79.center) to (78.center);
		\draw (80.center) to (65.center);
		\draw [bend left=165] (81.center) to (83.center);
		\draw (75.center) to (72.center);
		\draw [style=Red] (75.center) to (83.center);
		\draw [style=Red] (76.center) to (85.center);
		\draw [style=Yellow] (81.center) to (80.center);
		\draw [style=Yellow] (79.center) to (82.center);
		\draw (96) to (97.center);
		\draw (98) to (99.center);
		\draw (102.center) to (95.center);
		\draw (104.center) to (103.center);
		\draw (54) to (107);
		\draw (107) to (109);
		\draw (110) to (13);
		\draw (108) to (106);
		\draw (106) to (15);
		\draw (57) to (112);
		\draw (112) to (115);
		\draw (110) to (113);
		\draw [style=dottededge] (113) to (109);
		\draw [style=dottededge] (115) to (116.center);
		\draw [style=dottededge] (108) to (117.center);
		\draw [style=Red] (54) to (15);
		\draw [style=Pink] (107) to (106);
		\draw [style=Blue] (109) to (108);
		\draw [style=Turq] (113) to (115);
		\draw [style=Green] (112) to (110);
		\draw [style=Yellow] (13) to (57);
		\draw [style=Pink] (77.center) to (118.center);
		\draw (118.center) to (85.center);
		\draw [style=Pink] (120.center) to (128.center);
		\draw (129.center) to (128.center);
		\draw [style=Blue] (129.center) to (121.center);
		\draw (121.center) to (120.center);
		\draw [style=Blue] (132.center) to (134.center);
		\draw (134.center) to (135.center);
		\draw (133.center) to (132.center);
		\draw [style=dottededge] (74.center) to (135.center);
		\draw [style=Green] (78.center) to (123.center);
		\draw (123.center) to (82.center);
		\draw [style=Green] (140.center) to (138.center);
		\draw (138.center) to (139.center);
		\draw [style=Turq] (139.center) to (141.center);
		\draw (141.center) to (140.center);
		\draw (145.center) to (144.center);
		\draw [style=Turq] (144.center) to (124.center);
		\draw (124.center) to (125.center);
		\draw [style=dottededge] (145.center) to (73.center);
		\draw [style=dottededge] (125.center) to (133.center);
	\end{pgfonlayer}
\end{tikzpicture}
    \end{equation}
    \caption{CPM-like Interference diagram for a cyclic permutation between left and right wires is $\frac{1}{d^N}\pazocal{I}$ multiplied by $N-2$ closed loops, giving $\pazocal{N}_{\pi \pi'} = \frac{1}{d^2} \pazocal{I}$}
    \label{fig:cpi}
\end{figure}
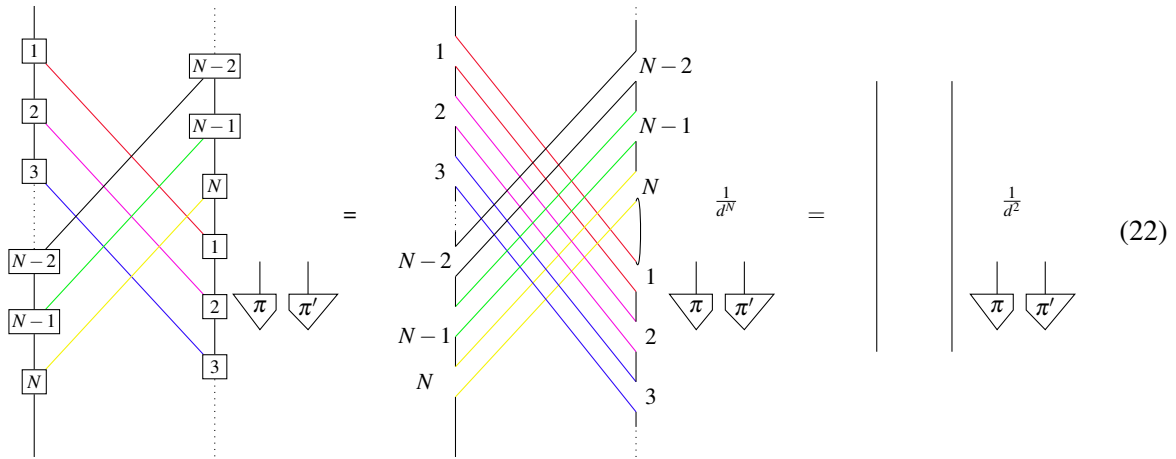
Each interference term between cyclic permutations is an information transmitting term with $N-2$ closed loops (the parallelograms in the right-hand-side of the first equality), ensuring that the prefactor of such a term is always $\frac{1}{d^2}$. Since $\pazocal{N}_{\pi \pi'}(\rho)$ gives $\frac{\rho}{d^2}$ when $\pi \neq \pi'$ and $\frac{I}{d}$ when $\pi = \pi'$ the superposition of the $N$ cyclic permutations of $N$ channels has a simple algebraic expression
\begin{equation}
    S_{NC}(\{ \pazocal{N}^{D(i)} \}_{i = 1}^{i = N})(\rho ,\ket{+}\bra{+}) = \sum_{\pi}  \frac{I}{Nd}\trace (\rho) \otimes \ket{\pi}\bra{\pi}  +
    \sum_{\pi \neq \pi'} \frac{\rho}{Nd^2} \otimes \ket{\pi}\bra{\pi'}
\end{equation}
For a general choice of $M$ permutations $\{\pi\}$, it is not true that for all $\pi,\pi'$ $\pazocal{N}_{\pi \pi'}(\rho) \propto  \rho$, and even when there exist $\pi,\pi'$ with $\pazocal{N}_{\pi \pi'}(\rho) \propto \rho$ it is often true that $\trace[\pazocal{N}_{\pi \pi'}] < \frac{1}{d^2}$. 

Guided by an intuition that the capacity enhancement will be greatest when the number and normalisation of information transmitting terms $\pazocal{N}_{\pi \pi'}$ is maximised, we expect the output channel for a general superposition of permutations would then have lower capacity than the case for which all $\pi,\pi'$ are mutually cyclic permutations, this intuition correctly predicts the highest capacity superpositions of $N=3$ CDPCs as explored in \cite{Procopio2020SendingOrders}. Letting the number of terms $\pazocal{N}_{\pi \pi'}$ proportional to the identity channel and the completely depolarising channel be $n_{Id}$ and $n_{Dp}$ respectively, for maximal capacity enhancement we suggest choosing the $M$ orders which optimize 
\begin{equation}
    \mathcal{O}(S) \equiv \frac{n_{id}E_{Id}}{n_{Dp}E_{Dp}}
\end{equation}
where $E_{Id}$, $E_{Dp}$ are the expected value of the normalisation $E(\trace[\pazocal{N}_{\pi \pi'}])$ for identity and depolarising terms.
In section 5.3, we give a general expression for any interference term $\pazocal{N}_{\pi \pi'}$, proving that cyclic permutations uniquely maximise $\mathcal{O}(S)$ for fixed $M \leq N$. This suggests that good candidates for high capacity activation of $N$ CDPCs given a quMit control should consist of superpositions of mutually cyclic permutations.

\subsection{Characterisation by Permutation Properties}


Generalising beyond the cyclic case we derive a simple condition on the permutations $\pi$ and $\pi'$ which can be used to completely determine any $\pazocal{N}_{\pi \pi'}$. Firstly ${\pi}$ and ${\pi'}$ can be used to define cycle permutations $C_\pi$ and $C_{\pi'}$ by 
\begin{equation}
    C_{\pi} \equiv  (0\pi(N)\pi(N-1)\dots \pi(1))
\end{equation}
\begin{equation}
    C_{\pi'} \equiv (0\pi'(1) \dots \pi'(N-1)\pi'(N))
\end{equation}
We will show that the interference diagram for $\pazocal{N}_{\pi \pi'}$ can be used to compute the product
\begin{equation}
    C_{\pi \pi'} \equiv C_{\pi'}^{-1} \circ C_{\pi} =  (0\pi'(N)\pi'(N-1)\dots \pi'(1))(0\pi(1) \dots \pi(N-1)\pi(N))
\end{equation}
and crucially we show the converse, that any term $\pazocal{N}_{\pi \pi'}$ can be computed by finding the cycle decomposition of $C_{\pi \pi'}$. As a corollary we will have demonstrated that the $\pazocal{N}_{\pi \pi'}$ are characterised by a cds sortability \cite{Adamyk2017SORTINGCYCLES} condition between $\pi$ and $\pi^{'}$. We write $c_{\pi \pi'}$ for the number of cycles in the cycle decomposition of $C_{\pi \pi'}$, $\pazocal{D}$ for the completely depolarising channel, and $\pazocal{I}$ for the identity channel. 
\begin{thm}[Information Transmission by Cycle Decomposition]
For the term $\pazocal{N}_{\pi \pi'}$ in the quantum switch of $M$ orders of $N$ CDPCs, 

\begin{itemize}
    \item $\pazocal{N}_{\pi \pi'} \propto $ $d\pazocal{D}$  $\Longleftrightarrow$
    $0$,$\pi(N)$ are not in the same cycle of $C_{\pi \pi'}$ 
    \item $\pazocal{N}_{\pi \pi'} \propto \pazocal{I}$ 
    $\Longleftrightarrow$ $0$,$\pi(N)$ are in the same cycle of $C_{\pi \pi'}$
\end{itemize}
\end{thm}
\begin{proof}

The permutation
\begin{equation}
    C_{\pi \pi'} \equiv C_{\pi'}^{-1} \circ C_{\pi} =  (0\pi'(N)\pi'(N-1)\dots \pi'(1))(0\pi(1) \dots \pi(N-1)\pi(N))
\end{equation}
Is the function $C_{\pi \pi'}(\pi(a)) = \pi^{'}({\pi'}^{-1}(\pi(a+1)) - 1)$. 
Using the CPM-like diagram for $\pazocal{N}_{\pi \pi'}$ in figure \ref{cpi1}, the following steps compute $C_{\pi \pi'}(\pi(a))$ by following a connected path along the diagram from label $\pi(a)$ at position $a$. 
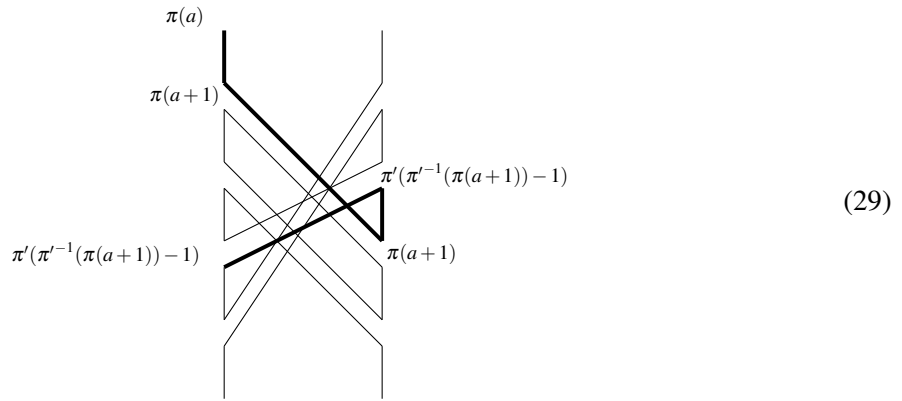
\begin{figure}[H]
    \centering
    \begin{equation}
        \begin{tikzpicture}[scale=0.7, {every node/.style}={scale=0.7}]
	\begin{pgfonlayer}{nodelayer}
		\node [style=none] (202) at (-1.5, 3.75) {};
		\node [style=none] (203) at (1.5, 3.75) {};
		\node [style=none] (204) at (-1.5, 3.25) {};
		\node [style=none] (205) at (1.5, 3.25) {};
		\node [style=none] (206) at (-1.5, 2.25) {};
		\node [style=none] (207) at (1.5, 2.25) {};
		\node [style=none] (208) at (-1.5, 1.75) {};
		\node [style=none] (209) at (1.5, 1.75) {};
		\node [style=none] (210) at (-1.5, 0.75) {};
		\node [style=none] (211) at (1.5, 0.75) {};
		\node [style=none] (212) at (-1.5, 0.25) {};
		\node [style=none] (213) at (1.5, 0.25) {};
		\node [style=none] (214) at (-1.5, -0.75) {};
		\node [style=none] (215) at (1.5, -0.75) {};
		\node [style=none] (216) at (-1.5, -1.25) {};
		\node [style=none] (217) at (1.5, -1.25) {};
		\node [style=none] (218) at (-1.5, -2.25) {};
		\node [style=none] (219) at (1.5, -2.25) {};
		\node [style=none] (220) at (-1.5, -2.75) {};
		\node [style=none] (221) at (1.5, -2.75) {};
		\node [style=none] (222) at (-1.5, -3.75) {};
		\node [style=none] (223) at (1.5, -3.75) {};
		\node [style=none] (224) at (-1.5, -4.25) {};
		\node [style=none] (225) at (1.5, -4.25) {};
		\node [style=none] (256) at (-2.25, 2) {$\pi(a+1)$};
		\node [style=none] (257) at (-1.75, 0.5) {};
		\node [style=none] (258) at (-1.75, -1) {};
		\node [style=none] (259) at (-1.75, -2.5) {};
		\node [style=none] (260) at (1.75, 2) {};
		\node [style=none] (261) at (3.25, 0.5) {$\pi'({\pi'}^{-1}(\pi(a+1)) -1)$};
		\node [style=none] (262) at (2.25, -1) {$\pi(a+1)$};
		\node [style=none] (263) at (1.75, -2.5) {};
		\node [style=none] (264) at (-2.25, 3.5) {$\pi(a)$};
		\node [style=none] (265) at (-3.75, -1) {$\pi'({\pi'}^{-1}(\pi(a+1)) -1)$};
	\end{pgfonlayer}
	\begin{pgfonlayer}{edgelayer}
		\draw [style=double edge] (206.center) to (204.center);
		\draw (205.center) to (207.center);
		\draw (209.center) to (211.center);
		\draw (208.center) to (210.center);
		\draw (212.center) to (214.center);
		\draw [style=double edge] (213.center) to (215.center);
		\draw (216.center) to (218.center);
		\draw (217.center) to (219.center);
		\draw (221.center) to (223.center);
		\draw (222.center) to (220.center);
		\draw [style=double edge] (206.center) to (215.center);
		\draw (208.center) to (217.center);
		\draw [style=double edge] (213.center) to (216.center);
		\draw (214.center) to (211.center);
		\draw (220.center) to (209.center);
		\draw (218.center) to (207.center);
		\draw (210.center) to (219.center);
		\draw (221.center) to (212.center);
	\end{pgfonlayer}
\end{tikzpicture}
    \end{equation}
    \caption{CPM-like diagram used to implement $C_{\pi \pi'}$. Being located $a$ slots from the top of the diagram, label $\pi_{a}$ is at position $a$ on the LHS}
    \label{cpi1}
\end{figure}
\begin{itemize}
    \item Start with label $\pi(a)$ at position $a$ on the LHS
    \item Move downwards to label $\pi(a+1)$ at position $a+1$
    \item Use wire to move to the same label $\pi(a+1)$ on the RHS, this will be at position ${\pi'}^{-1}(\pi(a+1))$
    \item Move up to the label $\pi'(\pi^{-1}(\pi(a+1)) - 1)$ at position ${\pi'}^{-1}(\pi(a+1))-1$  on the RHS 
    \item Use wire to move to move to same label $\pi'(\pi^{-1}(\pi(a+1)) - 1)$ on the LHS 
\end{itemize}
These steps implement $C_{\pi \pi'}(\pi(a))$ except for when the steps require a path which is undefined due to the open ends of the CPM-like digram, I.E when $\pi(a+1) = 0$ or $a=N$. The modification in figure \ref{cpi2} of the CPM-like diagram accounts for these edge cases and so can be used to compute $C_{\pi \pi'}(\pi(a))$ for any $a$.
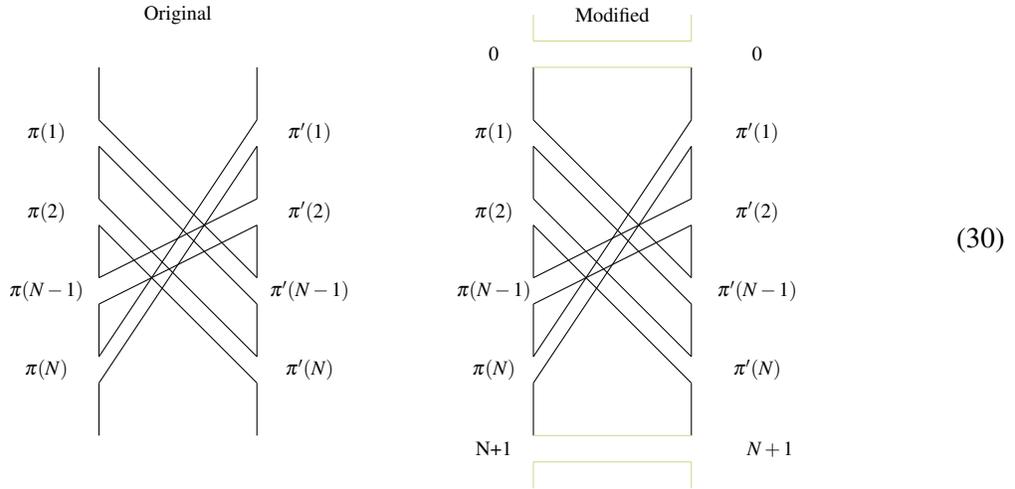
\begin{figure}
    \centering
    \begin{equation}
        \begin{tikzpicture}[scale=0.7, {every node/.style}={scale=0.7}]
	\begin{pgfonlayer}{nodelayer}
		\node [style=none] (204) at (-5.5, 3.25) {};
		\node [style=none] (205) at (-2.5, 3.25) {};
		\node [style=none] (206) at (-5.5, 2.25) {};
		\node [style=none] (207) at (-2.5, 2.25) {};
		\node [style=none] (208) at (-5.5, 1.75) {};
		\node [style=none] (209) at (-2.5, 1.75) {};
		\node [style=none] (210) at (-5.5, 0.75) {};
		\node [style=none] (211) at (-2.5, 0.75) {};
		\node [style=none] (212) at (-5.5, 0.25) {};
		\node [style=none] (213) at (-2.5, 0.25) {};
		\node [style=none] (214) at (-5.5, -0.75) {};
		\node [style=none] (215) at (-2.5, -0.75) {};
		\node [style=none] (216) at (-5.5, -1.25) {};
		\node [style=none] (217) at (-2.5, -1.25) {};
		\node [style=none] (218) at (-5.5, -2.25) {};
		\node [style=none] (219) at (-2.5, -2.25) {};
		\node [style=none] (220) at (-5.5, -2.75) {};
		\node [style=none] (221) at (-2.5, -2.75) {};
		\node [style=none] (222) at (-5.5, -3.75) {};
		\node [style=none] (223) at (-2.5, -3.75) {};
		\node [style=none] (257) at (-5.75, 0.5) {};
		\node [style=none] (258) at (-5.75, -1) {};
		\node [style=none] (259) at (-5.75, -2.5) {};
		\node [style=none] (260) at (-2.25, 2) {};
		\node [style=none] (263) at (-2.25, -2.5) {};
		\node [style=none] (265) at (-1.5, 2) {$\pi'(1)$};
		\node [style=none] (266) at (-1.5, 0.5) {$\pi'(2)$};
		\node [style=none] (267) at (-1.5, -1) {$\pi'(N-1)$};
		\node [style=none] (268) at (-1.5, -2.5) {$\pi'(N)$};
		\node [style=none] (279) at (2.75, 3.75) {};
		\node [style=none] (280) at (5.75, 3.75) {};
		\node [style=none] (281) at (2.75, 3.25) {};
		\node [style=none] (282) at (5.75, 3.25) {};
		\node [style=none] (283) at (2.75, 2.25) {};
		\node [style=none] (284) at (5.75, 2.25) {};
		\node [style=none] (285) at (2.75, 1.75) {};
		\node [style=none] (286) at (5.75, 1.75) {};
		\node [style=none] (287) at (2.75, 0.75) {};
		\node [style=none] (288) at (5.75, 0.75) {};
		\node [style=none] (289) at (2.75, 0.25) {};
		\node [style=none] (290) at (5.75, 0.25) {};
		\node [style=none] (291) at (2.75, -0.75) {};
		\node [style=none] (292) at (5.75, -0.75) {};
		\node [style=none] (293) at (2.75, -1.25) {};
		\node [style=none] (294) at (5.75, -1.25) {};
		\node [style=none] (295) at (2.75, -2.25) {};
		\node [style=none] (296) at (5.75, -2.25) {};
		\node [style=none] (297) at (2.75, -2.75) {};
		\node [style=none] (298) at (5.75, -2.75) {};
		\node [style=none] (299) at (2.75, -3.75) {};
		\node [style=none] (300) at (5.75, -3.75) {};
		\node [style=none] (301) at (2.75, -4.25) {};
		\node [style=none] (302) at (5.75, -4.25) {};
		\node [style=none] (304) at (2.5, 0.5) {};
		\node [style=none] (305) at (2.5, -1) {};
		\node [style=none] (306) at (2.5, -2.5) {};
		\node [style=none] (307) at (6, 2) {};
		\node [style=none] (308) at (6, -2.5) {};
		\node [style=none] (309) at (2, 3.5) {0};
		\node [style=none] (314) at (7, 3.5) {$0$};
		\node [style=none] (318) at (2.75, 4.25) {};
		\node [style=none] (319) at (5.75, 4.25) {};
		\node [style=none] (320) at (2.75, -4.75) {};
		\node [style=none] (321) at (5.75, -4.75) {};
		\node [style=none] (322) at (2, -4) {N+1};
		\node [style=none] (323) at (7.25, -4) {$N+1$};
		\node [style=none] (324) at (-4, 4.25) {Original};
		\node [style=none] (325) at (4.25, 4.25) {Modified};
		\node [style=none] (326) at (-6.5, 2) {$\pi(1)$};
		\node [style=none] (327) at (-6.5, 0.5) {$\pi(2)$};
		\node [style=none] (328) at (-6.5, -1) {$\pi(N-1)$};
		\node [style=none] (329) at (-6.5, -2.5) {$\pi(N)$};
		\node [style=none] (330) at (7, 2) {$\pi'(1)$};
		\node [style=none] (331) at (7, 0.5) {$\pi'(2)$};
		\node [style=none] (332) at (7, -1) {$\pi'(N-1)$};
		\node [style=none] (333) at (7, -2.5) {$\pi'(N)$};
		\node [style=none] (334) at (2, 2) {$\pi(1)$};
		\node [style=none] (335) at (2, 0.5) {$\pi(2)$};
		\node [style=none] (336) at (2, -1) {$\pi(N-1)$};
		\node [style=none] (337) at (2, -2.5) {$\pi(N)$};
	\end{pgfonlayer}
	\begin{pgfonlayer}{edgelayer}
		\draw (206.center) to (204.center);
		\draw (205.center) to (207.center);
		\draw (209.center) to (211.center);
		\draw (208.center) to (210.center);
		\draw (212.center) to (214.center);
		\draw (213.center) to (215.center);
		\draw (216.center) to (218.center);
		\draw (217.center) to (219.center);
		\draw (221.center) to (223.center);
		\draw (222.center) to (220.center);
		\draw (206.center) to (215.center);
		\draw (208.center) to (217.center);
		\draw (213.center) to (216.center);
		\draw (214.center) to (211.center);
		\draw (220.center) to (209.center);
		\draw (218.center) to (207.center);
		\draw (210.center) to (219.center);
		\draw (221.center) to (212.center);
		\draw (283.center) to (281.center);
		\draw (282.center) to (284.center);
		\draw (286.center) to (288.center);
		\draw (285.center) to (287.center);
		\draw (289.center) to (291.center);
		\draw (290.center) to (292.center);
		\draw (293.center) to (295.center);
		\draw (294.center) to (296.center);
		\draw (298.center) to (300.center);
		\draw (299.center) to (297.center);
		\draw (283.center) to (292.center);
		\draw (285.center) to (294.center);
		\draw (290.center) to (293.center);
		\draw (291.center) to (288.center);
		\draw (297.center) to (286.center);
		\draw (295.center) to (284.center);
		\draw (287.center) to (296.center);
		\draw (298.center) to (289.center);
		\draw [style=Beige] (320.center) to (301.center);
		\draw [style=Beige] (301.center) to (302.center);
		\draw [style=Beige] (302.center) to (321.center);
		\draw [style=Beige] (279.center) to (280.center);
		\draw [style=Beige] (280.center) to (319.center);
		\draw [style=Beige] (279.center) to (318.center);
		\draw [style=Beige] (281.center) to (282.center);
		\draw [style=Beige] (299.center) to (300.center);
	\end{pgfonlayer}
\end{tikzpicture}
    \end{equation}
    \caption{Modification of CPM-like diagram so that it may be used to fully evaluate $C_{\pi \pi'}$}
    \label{cpi2}
\end{figure}
If by starting at label $0$ on the LHS, iterating the above steps reaches node $\pi(N)$ on the LHS, the diagram is connected from top left to the bottom left, the same will be true for the RHS since the unmodified diagram can have no other open ends, and the diagram will be proportional to the identity channel. Iteration of the above steps is repeated application of $C_{\pi \pi'}$, it follows that if in the cycle decomposition of $C_{\pi \pi'}$, $0$ and $\pi(N)$ are in the same cycle, then $\pazocal{N}_{\pi \pi'}$ is proportional to the identity channel. Alternatively if $0$ and $\pi(N)$ are not in the same cycle the channel is proportional to the completely depolarising channel.
\end{proof}
Furthermore the cycle decomposition of $C_{\pi \pi'}$ completely determines the normalisation of each $\pazocal{N}_{\pi \pi'}$.

\begin{thm}[Normalisation by Cycle Decomposition]
The normalisation of $\pazocal{N}_{\pi \pi'}$ is determined by the number of cycles in the cycle decomposition of $C_{\pi \pi'}$
\begin{itemize}
    \item $\pazocal{N}_{\pi \pi'} = \frac{1}{d^N}d^{c_{\pi \pi'} - 2}d \pazocal{D}$ when $0$ and $\pi(N)$ are not in the same cycle of $C_{\pi \pi'}$
    \item $\pazocal{N}_{\pi \pi'} = \frac{1}{d^N}d^{c_{\pi \pi'} - 1} \pazocal{I}$ when when $0$ and $\pi(N)$ are in the same cycle of $C_{\pi \pi'}$
\end{itemize}
\end{thm}
\begin{proof}
    Given in Appendix A. 
\end{proof}
In \cite{Adamyk2017SORTINGCYCLES} it is proved that $\pi$ and $\pi^{'}$ are cds sortable iff $0$ and $\pi^{'}(N)$ are in the same cycle of $C_{\pi \pi'}$. 
\begin{coro}[Information Transmission by CDS Sortability] The term $\pazocal{N}_{\pi \pi'}$ is
\begin{itemize}
    \item Proportional to the completely depolarising channel if $\pi$ and $\pi^{'}$ are cds sortable
    \item Proportional to the Identity channel if $\pi$ and $\pi^{'}$ are not cds sortable
\end{itemize}
\end{coro}

It is also shown in \cite{Adamyk2017SORTINGCYCLES} that for any $\pi,\pi'$ with $\pi$ not cds sortable to $\pi^{'}$, $C_{\pi \pi'}$ has maximal number of cycles in its cycle decomposition if and only if $\pi^{'}\pi^{-1}$ is a cyclic permutation.

\begin{coro}[Optimising $\mathcal{O}(S)$]
The Cyclic Permutation Protocol Optimises $\mathcal{O}(S)$ for $M \leq N$
\end{coro}
\begin{proof}
    Given in Appendix B.
\end{proof}

\begin{coro}[Generic Output For the Switch of Depolarising Channels]
The generic output of the quantum switch of $M$ orders of $N$ completely depolarising channels may be written in terms of the set \textbf{cds} of pairs of permutations which are sortable to each-other by cds permutations
\begin{equation}
\begin{split}
\rho'  = \frac{1}{M}\sum_{(\pi,\pi') \in \textbf{cds}}  d^{(c_{\pi \pi'} - 1 - N)} \pazocal{D} \otimes \ket{\pi}\bra{\pi'} + \frac{1}{M}\sum_{(\pi,\pi') \not\in \textbf{cds}}  d^{(c_{\pi \pi'} - 1 - N)} \pazocal{I} \otimes \ket{\pi}\bra{\pi'}
\end{split}
\end{equation}
\end{coro}

\section{Summary}
By translating an algebraic expression for the quantum switch into a sum of diagrams we have found an alternative intuition for capacity activation in terms of quantum teleportation, and generalised computational results to the $N$ cyclic permutations of $N$ channels. We further derived a condition for separability of a diagram in terms of cds sortability which we in turn used to demonstrate optimality of the cyclic protocol with respect to a simple heuristic for information transmission. 
\section{Acknowledgments}
GC acknowledges a stimulating discussion with  S Popescu and P Skrzypczyk, who originally devised  Eq. (\ref{teleportation}).  We thank J Hefford, H Kristjánsson, and A Vanrietvelde for useful discussions. This work was supported by the National Natural Science Foundation of China through grant 11675136, by the Hong Kong Research Grant Council through grant 17307719 and though the Senior Research Fellowship Scheme SRFS2021-7S02, by the Croucher Foundation, and by the John Templeton Foundation through grant 61466, The Quantum Information Structure of Spacetime (qiss.fr). Research at the Perimeter Institute is supported by the Government of Canada through the Department of Innovation, Science and Economic Development Canada and by the Province of Ontario through the Ministry of Research, Innovation and Science. The opinions expressed in this publication are those of the authors and do not necessarily reflect the views of the John Templeton Foundation. MW acknowledges support by University College London and the EPSRC Doctoral Training Centre for Delivering Quantum Technologies.


\begin{thebibliography}{10}
\providecommand{\bibitemdeclare}[2]{}
\providecommand{\surnamestart}{}
\providecommand{\surnameend}{}
\providecommand{\urlprefix}{Available at }
\providecommand{\url}[1]{\texttt{#1}}
\providecommand{\href}[2]{\texttt{#2}}
\providecommand{\urlalt}[2]{\href{#1}{#2}}
\providecommand{\doi}[1]{doi:\urlalt{http://dx.doi.org/#1}{#1}}
\providecommand{\bibinfo}[2]{#2}
\providecommand{\arxiv}[1]{arXiv:\urlalt{http://arxiv.org/abs/#1}{#1}}

\bibitemdeclare{article}{Abbott2018CommunicationChannels}
\bibitem{Abbott2018CommunicationChannels}
\bibinfo{author}{Alastair~A. \surnamestart Abbott\surnameend},
  \bibinfo{author}{Julian \surnamestart Wechs\surnameend},
  \bibinfo{author}{Dominic \surnamestart Horsman\surnameend},
  \bibinfo{author}{Mehdi \surnamestart Mhalla\surnameend} \&
  \bibinfo{author}{Cyril \surnamestart Branciard\surnameend}
  (\bibinfo{year}{2020}): \emph{\bibinfo{title}{{Communication through coherent
  control of quantum channels}}}.
\newblock {\sl \bibinfo{journal}{Quantum}} \bibinfo{volume}{4}, p.
  \bibinfo{pages}{333}, \doi{10.22331/Q-2020-09-24-333}.

\bibitemdeclare{inproceedings}{Abramsky2004AProtocols}
\bibitem{Abramsky2004AProtocols}
\bibinfo{author}{Samson \surnamestart Abramsky\surnameend} \&
  \bibinfo{author}{Bob \surnamestart Coecke\surnameend} (\bibinfo{year}{2004}):
  \emph{\bibinfo{title}{{A categorical semantics of quantum protocols}}}.
\newblock In: {\sl \bibinfo{booktitle}{Proceedings - Symposium on Logic in
  Computer Science}}, \bibinfo{volume}{19}, pp. \bibinfo{pages}{415--425},
  \doi{10.1109/lics.2004.1319636}.

\bibitemdeclare{article}{Adamyk2017SORTINGCYCLES}
\bibitem{Adamyk2017SORTINGCYCLES}
\bibinfo{author}{K.~L.M. \surnamestart Adamyk\surnameend},
  \bibinfo{author}{E.~\surnamestart Holmes\surnameend}, \bibinfo{author}{G.~R.
  \surnamestart Mayfield\surnameend}, \bibinfo{author}{D.~J. \surnamestart
  Moritz\surnameend}, \bibinfo{author}{M.~\surnamestart Scheepers\surnameend},
  \bibinfo{author}{B.~E. \surnamestart Tenner\surnameend} \&
  \bibinfo{author}{H.~C. \surnamestart Wauck\surnameend}
  (\bibinfo{year}{2017}): \emph{\bibinfo{title}{{Sorting permutations: Games,
  genomes, and cycles}}}.
\newblock {\sl \bibinfo{journal}{Discrete Mathematics, Algorithms and
  Applications}} \bibinfo{volume}{9}(\bibinfo{number}{5}),
  \doi{10.1142/S179383091750063X}.

\bibitemdeclare{article}{Araujo2014ComputationalGates}
\bibitem{Araujo2014ComputationalGates}
\bibinfo{author}{Mateus \surnamestart Ara{\'{u}}jo\surnameend},
  \bibinfo{author}{Fabio \surnamestart Costa\surnameend} \&
  \bibinfo{author}{Časlav \surnamestart Brukner\surnameend}
  (\bibinfo{year}{2014}): \emph{\bibinfo{title}{{Computational advantage from
  quantum-controlled ordering of gates}}}.
\newblock {\sl \bibinfo{journal}{Physical Review Letters}}
  \bibinfo{volume}{113}(\bibinfo{number}{25}), p. \bibinfo{pages}{250402},
  \doi{10.1103/PhysRevLett.113.250402}.

\bibitemdeclare{article}{Boixo2012EntangledDephasing}
\bibitem{Boixo2012EntangledDephasing}
\bibinfo{author}{Sergio \surnamestart Boixo\surnameend} \&
  \bibinfo{author}{Chris \surnamestart Heunen\surnameend}
  (\bibinfo{year}{2012}): \emph{\bibinfo{title}{{Entangled and sequential
  quantum protocols with dephasing}}}.
\newblock {\sl \bibinfo{journal}{Physical Review Letters}}
  \bibinfo{volume}{108}(\bibinfo{number}{12}), p. \bibinfo{pages}{120402},
  \doi{10.1103/PhysRevLett.108.120402}.

\bibitemdeclare{article}{Chancellor2016GraphicalCorrection}
\bibitem{Chancellor2016GraphicalCorrection}
\bibinfo{author}{Nicholas \surnamestart Chancellor\surnameend},
  \bibinfo{author}{Aleks \surnamestart Kissinger\surnameend},
  \bibinfo{author}{Joschka \surnamestart Roffe\surnameend},
  \bibinfo{author}{Stefan \surnamestart Zohren\surnameend} \&
  \bibinfo{author}{Dominic \surnamestart Horsman\surnameend}
  (\bibinfo{year}{2016}): \emph{\bibinfo{title}{Graphical Structures for Design
  and Verification of Quantum Error Correction}}.
\newblock {\sl \bibinfo{journal}{\arxiv{21611.08012}}}.

\bibitemdeclare{article}{Chiribella2008TransformingSupermaps}
\bibitem{Chiribella2008TransformingSupermaps}
\bibinfo{author}{G.~\surnamestart Chiribella\surnameend},
  \bibinfo{author}{G.~M. \surnamestart D'Ariano\surnameend} \&
  \bibinfo{author}{P.~\surnamestart Perinotti\surnameend}
  (\bibinfo{year}{2008}): \emph{\bibinfo{title}{{Transforming quantum
  operations: Quantum supermaps}}}.
\newblock {\sl \bibinfo{journal}{EPL (Europhysics Letters)}}
  \bibinfo{volume}{83}(\bibinfo{number}{3}), p. \bibinfo{pages}{30004},
  \doi{10.1209/0295-5075/83/30004}.

\bibitemdeclare{article}{chiribella2009beyond}
\bibitem{chiribella2009beyond}
\bibinfo{author}{G~\surnamestart Chiribella\surnameend},
  \bibinfo{author}{GM~\surnamestart D’Ariano\surnameend},
  \bibinfo{author}{P~\surnamestart Perinotti\surnameend} \&
  \bibinfo{author}{B~\surnamestart Valiron\surnameend} (\bibinfo{year}{2009}):
  \emph{\bibinfo{title}{Beyond quantum computers}}.
\newblock {\sl \bibinfo{journal}{\arxiv{0912.0195}}}.

\bibitemdeclare{article}{Chiribella2012PerfectStructures}
\bibitem{Chiribella2012PerfectStructures}
\bibinfo{author}{Giulio \surnamestart Chiribella\surnameend}
  (\bibinfo{year}{2012}): \emph{\bibinfo{title}{{Perfect discrimination of
  no-signalling channels via quantum superposition of causal structures}}}.
\newblock {\sl \bibinfo{journal}{Physical Review A - Atomic, Molecular, and
  Optical Physics}} \bibinfo{volume}{86}(\bibinfo{number}{4}), p.
  \bibinfo{pages}{040301}, \doi{10.1103/PhysRevA.86.040301}.

\bibitemdeclare{article}{Chiribella2018IndefiniteChannel}
\bibitem{Chiribella2018IndefiniteChannel}
\bibinfo{author}{Giulio \surnamestart Chiribella\surnameend},
  \bibinfo{author}{Manik \surnamestart Banik\surnameend},
  \bibinfo{author}{Some~Sankar \surnamestart Bhattacharya\surnameend},
  \bibinfo{author}{Tamal \surnamestart Guha\surnameend}, \bibinfo{author}{Mir
  \surnamestart Alimuddin\surnameend}, \bibinfo{author}{Arup \surnamestart
  Roy\surnameend}, \bibinfo{author}{Sutapa \surnamestart Saha\surnameend},
  \bibinfo{author}{Sristy \surnamestart Agrawal\surnameend} \&
  \bibinfo{author}{Guruprasad \surnamestart Kar\surnameend}
  (\bibinfo{year}{2018}): \emph{\bibinfo{title}{Indefinite causal order enables
  perfect quantum communication with zero capacity channel}}.
\newblock {\sl \bibinfo{journal}{\arxiv{1810.10457}}}.

\bibitemdeclare{article}{chiribella2009theoretical}
\bibitem{chiribella2009theoretical}
\bibinfo{author}{Giulio \surnamestart Chiribella\surnameend},
  \bibinfo{author}{Giacomo~Mauro \surnamestart D'Ariano\surnameend} \&
  \bibinfo{author}{Paolo \surnamestart Perinotti\surnameend}
  (\bibinfo{year}{2009}): \emph{\bibinfo{title}{Theoretical framework for
  quantum networks}}.
\newblock {\sl \bibinfo{journal}{Phys. Rev. A}} \bibinfo{volume}{80}, p.
  \bibinfo{pages}{022339}, \doi{10.1103/PhysRevA.80.022339}.

\bibitemdeclare{article}{Chiribella2013QuantumStructure}
\bibitem{Chiribella2013QuantumStructure}
\bibinfo{author}{Giulio \surnamestart Chiribella\surnameend},
  \bibinfo{author}{Giacomo~Mauro \surnamestart D'Ariano\surnameend},
  \bibinfo{author}{Paolo \surnamestart Perinotti\surnameend} \&
  \bibinfo{author}{Benoit \surnamestart Valiron\surnameend}
  (\bibinfo{year}{2013}): \emph{\bibinfo{title}{{Quantum computations without
  definite causal structure}}}.
\newblock {\sl \bibinfo{journal}{Physical Review A - Atomic, Molecular, and
  Optical Physics}} \bibinfo{volume}{88}(\bibinfo{number}{2}), p.
  \bibinfo{pages}{022318}, \doi{10.1103/PhysRevA.88.022318}.

\bibitemdeclare{article}{chiribella2020quantumcomm}
\bibitem{chiribella2020quantumcomm}
\bibinfo{author}{Giulio \surnamestart Chiribella\surnameend},
  \bibinfo{author}{Matthew \surnamestart Wilson\surnameend} \&
  \bibinfo{author}{H.~F. \surnamestart Chau\surnameend} (\bibinfo{year}{2020}):
  \emph{\bibinfo{title}{Quantum and Classical Data Transmission Through
  Completely Depolarising Channels in a Superposition of Cyclic Orders}}.
\newblock {\sl \bibinfo{journal}{\arxiv{2005.00618}}}.

\bibitemdeclare{article}{Coecke2008AxiomaticCPM-construction}
\bibitem{Coecke2008AxiomaticCPM-construction}
\bibinfo{author}{Bob \surnamestart Coecke\surnameend} (\bibinfo{year}{2008}):
  \emph{\bibinfo{title}{{Axiomatic Description of Mixed States From Selinger's
  CPM-construction}}}.
\newblock {\sl \bibinfo{journal}{Electronic Notes in Theoretical Computer
  Science}} \bibinfo{volume}{210}(\bibinfo{number}{C}), pp.
  \bibinfo{pages}{3--13}, \doi{10.1016/j.entcs.2008.04.014}.

\bibitemdeclare{article}{Coecke2010QuantumPicturalism}
\bibitem{Coecke2010QuantumPicturalism}
\bibinfo{author}{Bob \surnamestart Coecke\surnameend} (\bibinfo{year}{2010}):
  \emph{\bibinfo{title}{{Quantum picturalism}}}.
\newblock {\sl \bibinfo{journal}{Contemporary Physics}}
  \bibinfo{volume}{51}(\bibinfo{number}{1}), pp. \bibinfo{pages}{59--83},
  \doi{10.1080/00107510903257624}.

\bibitemdeclare{article}{Coecke2013CausalProcesses}
\bibitem{Coecke2013CausalProcesses}
\bibinfo{author}{Bob \surnamestart Coecke\surnameend} \&
  \bibinfo{author}{Raymond \surnamestart Lal\surnameend}
  (\bibinfo{year}{2013}): \emph{\bibinfo{title}{{Causal Categories:
  Relativistically Interacting Processes}}}.
\newblock {\sl \bibinfo{journal}{Foundations of Physics}}
  \bibinfo{volume}{43}(\bibinfo{number}{4}), pp. \bibinfo{pages}{458--501},
  \doi{10.1007/s10701-012-9646-8}.

\bibitemdeclare{article}{colnaghi2012quantum}
\bibitem{colnaghi2012quantum}
\bibinfo{author}{Timoteo \surnamestart Colnaghi\surnameend},
  \bibinfo{author}{Giacomo~Mauro \surnamestart D'Ariano\surnameend},
  \bibinfo{author}{Stefano \surnamestart Facchini\surnameend} \&
  \bibinfo{author}{Paolo \surnamestart Perinotti\surnameend}
  (\bibinfo{year}{2012}): \emph{\bibinfo{title}{Quantum computation with
  programmable connections between gates}}.
\newblock {\sl \bibinfo{journal}{Physics Letters A}}
  \bibinfo{volume}{376}(\bibinfo{number}{45}), pp. \bibinfo{pages}{2940--2943},
  \doi{10.1016/j.physleta.2012.08.028}.

\bibitemdeclare{article}{Ebler2018EnhancedOrder}
\bibitem{Ebler2018EnhancedOrder}
\bibinfo{author}{Daniel \surnamestart Ebler\surnameend}, \bibinfo{author}{Sina
  \surnamestart Salek\surnameend} \& \bibinfo{author}{Giulio \surnamestart
  Chiribella\surnameend} (\bibinfo{year}{2018}):
  \emph{\bibinfo{title}{{Enhanced Communication with the Assistance of
  Indefinite Causal Order}}}.
\newblock {\sl \bibinfo{journal}{Physical Review Letters}}
  \bibinfo{volume}{120}(\bibinfo{number}{12}), p. \bibinfo{pages}{120502},
  \doi{10.1103/PhysRevLett.120.120502}.

\bibitemdeclare{inproceedings}{Facchini2015QuantumProblem}
\bibitem{Facchini2015QuantumProblem}
\bibinfo{author}{Stefano \surnamestart Facchini\surnameend} \&
  \bibinfo{author}{Simon \surnamestart Perdrix\surnameend}
  (\bibinfo{year}{2015}): \emph{\bibinfo{title}{{Quantum circuits for the
  unitary permutation problem}}}.
\newblock In: {\sl \bibinfo{booktitle}{Lecture Notes in Computer Science
  (including subseries Lecture Notes in Artificial Intelligence and Lecture
  Notes in Bioinformatics)}}, \bibinfo{volume}{9076},
  \bibinfo{publisher}{Springer Verlag}, pp. \bibinfo{pages}{324--331},
  \doi{10.1007/978-3-319-17142-5{\_}28}.

\bibitemdeclare{article}{Gogioso2017FullyProblem}
\bibitem{Gogioso2017FullyProblem}
\bibinfo{author}{Stefano \surnamestart Gogioso\surnameend} \&
  \bibinfo{author}{Aleks \surnamestart Kissinger\surnameend}
  (\bibinfo{year}{2017}): \emph{\bibinfo{title}{Fully graphical treatment of
  the quantum algorithm for the Hidden Subgroup Problem}}.
\newblock {\sl \bibinfo{journal}{\arxiv{1701.08669}}}.

\bibitemdeclare{article}{Guerin2016ExponentialCommunication}
\bibitem{Guerin2016ExponentialCommunication}
\bibinfo{author}{Philippe~Allard \surnamestart Gu{\'{e}}rin\surnameend},
  \bibinfo{author}{Adrien \surnamestart Feix\surnameend},
  \bibinfo{author}{Mateus \surnamestart Ara{\'{u}}jo\surnameend} \&
  \bibinfo{author}{Časlav \surnamestart Brukner\surnameend}
  (\bibinfo{year}{2016}): \emph{\bibinfo{title}{{Exponential Communication
  Complexity Advantage from Quantum Superposition of the Direction of
  Communication}}}.
\newblock {\sl \bibinfo{journal}{Physical Review Letters}}
  \bibinfo{volume}{117}(\bibinfo{number}{10}), p. \bibinfo{pages}{100502},
  \doi{10.1103/PhysRevLett.117.100502}.

\bibitemdeclare{article}{Guerin2019CommunicationNoise}
\bibitem{Guerin2019CommunicationNoise}
\bibinfo{author}{Philippe~Allard \surnamestart Gu{\'{e}}rin\surnameend},
  \bibinfo{author}{Giulia \surnamestart Rubino\surnameend} \&
  \bibinfo{author}{Časlav \surnamestart Brukner\surnameend}
  (\bibinfo{year}{2019}): \emph{\bibinfo{title}{{Communication through
  quantum-controlled noise}}}.
\newblock {\sl \bibinfo{journal}{Physical Review A}}
  \bibinfo{volume}{99}(\bibinfo{number}{6}), \doi{10.1103/PhysRevA.99.062317}.

\bibitemdeclare{article}{Holevo1998TheStatesb}
\bibitem{Holevo1998TheStatesb}
\bibinfo{author}{A.~S. \surnamestart Holevo\surnameend} (\bibinfo{year}{1998}):
  \emph{\bibinfo{title}{{The capacity of the quantum channel with general
  signal states}}}.
\newblock {\sl \bibinfo{journal}{IEEE Transactions on Information Theory}}
  \bibinfo{volume}{44}(\bibinfo{number}{1}), pp. \bibinfo{pages}{269--273},
  \doi{10.1109/18.651037}.

\bibitemdeclare{article}{Kissinger2017Picture-perfectDistribution}
\bibitem{Kissinger2017Picture-perfectDistribution}
\bibinfo{author}{Aleks \surnamestart Kissinger\surnameend},
  \bibinfo{author}{Sean \surnamestart Tull\surnameend} \& \bibinfo{author}{Bas
  \surnamestart Westerbaan\surnameend} (\bibinfo{year}{2017}):
  \emph{\bibinfo{title}{Picture-perfect Quantum Key Distribution}}.
\newblock {\sl \bibinfo{journal}{\arxiv{1704.08668}}}.

\bibitemdeclare{article}{Kissinger2019AStructure}
\bibitem{Kissinger2019AStructure}
\bibinfo{author}{Aleks \surnamestart Kissinger\surnameend} \&
  \bibinfo{author}{Sander \surnamestart Uijlen\surnameend}
  (\bibinfo{year}{2019}): \emph{\bibinfo{title}{{A categorical semantics for
  causal structure}}}.
\newblock {\sl \bibinfo{journal}{Logical Methods in Computer Science}}
  \bibinfo{volume}{15}(\bibinfo{number}{3}), \doi{10.23638/LMCS-15(3:15)2019}.

\bibitemdeclare{article}{Kristjansson2019ResourceProcesses}
\bibitem{Kristjansson2019ResourceProcesses}
\bibinfo{author}{Hlér \surnamestart Kristj{\'{a}}nsson\surnameend},
  \bibinfo{author}{Giulio \surnamestart Chiribella\surnameend},
  \bibinfo{author}{Sina \surnamestart Salek\surnameend},
  \bibinfo{author}{Daniel \surnamestart Ebler\surnameend} \&
  \bibinfo{author}{Matthew \surnamestart Wilson\surnameend}
  (\bibinfo{year}{2020}): \emph{\bibinfo{title}{{Resource theories of
  communication}}}.
\newblock {\sl \bibinfo{journal}{New Journal of Physics}}
  \bibinfo{volume}{22}(\bibinfo{number}{7}), p. \bibinfo{pages}{073014},
  \doi{10.1088/1367-2630/ab8ef7}.

\bibitemdeclare{article}{Mukhopadhyay2018SuperpositionThermometry}
\bibitem{Mukhopadhyay2018SuperpositionThermometry}
\bibinfo{author}{Chiranjib \surnamestart Mukhopadhyay\surnameend},
  \bibinfo{author}{Manish~K. \surnamestart Gupta\surnameend} \&
  \bibinfo{author}{Arun~Kumar \surnamestart Pati\surnameend}
  (\bibinfo{year}{2018}): \emph{\bibinfo{title}{Superposition of causal order
  as a metrological resource for quantum thermometry}}.
\newblock {\sl \bibinfo{journal}{\arxiv{1812.07508}}}.

\bibitemdeclare{article}{Oreshkov2012QuantumOrder}
\bibitem{Oreshkov2012QuantumOrder}
\bibinfo{author}{Ognyan \surnamestart Oreshkov\surnameend},
  \bibinfo{author}{Fabio \surnamestart Costa\surnameend} \&
  \bibinfo{author}{Časlav \surnamestart Brukner\surnameend}
  (\bibinfo{year}{2012}): \emph{\bibinfo{title}{{Quantum correlations with no
  causal order}}}.
\newblock {\sl \bibinfo{journal}{Nature Communications}} \bibinfo{volume}{3},
  \doi{10.1038/ncomms2076}.

\bibitemdeclare{article}{Procopio2019CommunicationScenario}
\bibitem{Procopio2019CommunicationScenario}
\bibinfo{author}{Lorenzo~M. \surnamestart Procopio\surnameend},
  \bibinfo{author}{Francisco \surnamestart Delgado\surnameend},
  \bibinfo{author}{Marco \surnamestart Enr{\'{i}}quez\surnameend},
  \bibinfo{author}{Nadia \surnamestart Belabas\surnameend} \&
  \bibinfo{author}{Juan~Ariel \surnamestart Levenson\surnameend}
  (\bibinfo{year}{2019}): \emph{\bibinfo{title}{{Communication Enhancement
  through Quantum Coherent Control of N Channels in an Indefinite Causal-Order
  Scenario}}}.
\newblock {\sl \bibinfo{journal}{Entropy}}
  \bibinfo{volume}{21}(\bibinfo{number}{10}), p. \bibinfo{pages}{1012},
  \doi{10.3390/e21101012}.

\bibitemdeclare{article}{Procopio2020SendingOrders}
\bibitem{Procopio2020SendingOrders}
\bibinfo{author}{Lorenzo~M. \surnamestart Procopio\surnameend},
  \bibinfo{author}{Francisco \surnamestart Delgado\surnameend},
  \bibinfo{author}{Marco \surnamestart Enr{\'{i}}quez\surnameend},
  \bibinfo{author}{Nadia \surnamestart Belabas\surnameend} \&
  \bibinfo{author}{Juan~Ariel \surnamestart Levenson\surnameend}
  (\bibinfo{year}{2020}): \emph{\bibinfo{title}{{Sending classical information
  via three noisy channels in superposition of causal orders}}}.
\newblock {\sl \bibinfo{journal}{Physical Review A}}
  \bibinfo{volume}{101}(\bibinfo{number}{1}), p. \bibinfo{pages}{012346},
  \doi{10.1103/PhysRevA.101.012346}.

\bibitemdeclare{article}{Ranchin2014CompleteMechanics}
\bibitem{Ranchin2014CompleteMechanics}
\bibinfo{author}{André \surnamestart Ranchin\surnameend} \&
  \bibinfo{author}{Bob \surnamestart Coecke\surnameend} (\bibinfo{year}{2014}):
  \emph{\bibinfo{title}{{Complete set of circuit equations for stabilizer
  quantum mechanics}}}.
\newblock {\sl \bibinfo{journal}{Physical Review A - Atomic, Molecular, and
  Optical Physics}} \bibinfo{volume}{90}(\bibinfo{number}{1}), p.
  \bibinfo{pages}{012109}, \doi{10.1103/PhysRevA.90.012109}.

\bibitemdeclare{article}{Salek2018QuantumOrders}
\bibitem{Salek2018QuantumOrders}
\bibinfo{author}{Sina \surnamestart Salek\surnameend}, \bibinfo{author}{Daniel
  \surnamestart Ebler\surnameend} \& \bibinfo{author}{Giulio \surnamestart
  Chiribella\surnameend} (\bibinfo{year}{2018}): \emph{\bibinfo{title}{Quantum
  communication in a superposition of causal orders}}.
\newblock {\sl \bibinfo{journal}{\arxiv{1809.06655}}}.

\bibitemdeclare{article}{sazim2020classical}
\bibitem{sazim2020classical}
\bibinfo{author}{Sk~\surnamestart Sazim\surnameend}, \bibinfo{author}{Kratveer
  \surnamestart Singh\surnameend} \& \bibinfo{author}{Arun~Kumar \surnamestart
  Pati\surnameend} (\bibinfo{year}{2020}): \emph{\bibinfo{title}{Classical
  Communications with Indefinite Causal Order for $N$ completely depolarizing
  channels}}.
\newblock {\sl \bibinfo{journal}{\arxiv{2004.14339}}}.

\bibitemdeclare{article}{Schumacher1995QuantumCoding}
\bibitem{Schumacher1995QuantumCoding}
\bibinfo{author}{Benjamin \surnamestart Schumacher\surnameend}
  (\bibinfo{year}{1995}): \emph{\bibinfo{title}{{Quantum coding}}}.
\newblock {\sl \bibinfo{journal}{Physical Review A}}
  \bibinfo{volume}{51}(\bibinfo{number}{4}), pp. \bibinfo{pages}{2738--2747},
  \doi{10.1103/PhysRevA.51.2738}.

\bibitemdeclare{article}{Selinger2007DaggerAbstract}
\bibitem{Selinger2007DaggerAbstract}
\bibinfo{author}{Peter \surnamestart Selinger\surnameend}
  (\bibinfo{year}{2007}): \emph{\bibinfo{title}{{Dagger Compact Closed
  Categories and Completely Positive Maps. (Extended Abstract)}}}.
\newblock {\sl \bibinfo{journal}{Electronic Notes in Theoretical Computer
  Science}} \bibinfo{volume}{170}, pp. \bibinfo{pages}{139--163},
  \doi{10.1016/j.entcs.2006.12.018}.

\bibitemdeclare{article}{Vicary2012TheAlgorithms}
\bibitem{Vicary2012TheAlgorithms}
\bibinfo{author}{Jamie \surnamestart Vicary\surnameend} (\bibinfo{year}{2012}):
  \emph{\bibinfo{title}{{The Topology of Quantum Algorithms}}}.
\newblock {\sl \bibinfo{journal}{Proceedings - Symposium on Logic in Computer
  Science}}, pp. \bibinfo{pages}{93--102}, \doi{10.1109/LICS.2013.14}.

\bibitemdeclare{book}{Wilde2013QuantumTheory}
\bibitem{Wilde2013QuantumTheory}
\bibinfo{author}{Mark~M. \surnamestart Wilde\surnameend}
  (\bibinfo{year}{2013}): \emph{\bibinfo{title}{{Quantum Information Theory}}}.
\newblock \bibinfo{publisher}{Cambridge University Press},
  \bibinfo{address}{Cambridge}, \doi{10.1017/CBO9781139525343}.

\bibitemdeclare{inproceedings}{Zeng2014AbstractAlgorithms}
\bibitem{Zeng2014AbstractAlgorithms}
\bibinfo{author}{William \surnamestart Zeng\surnameend} \&
  \bibinfo{author}{Jamie \surnamestart Vicary\surnameend}
  (\bibinfo{year}{2014}): \emph{\bibinfo{title}{{Abstract structure of unitary
  oracles for quantum algorithms}}}.
\newblock In: {\sl \bibinfo{booktitle}{Electronic Proceedings in Theoretical
  Computer Science, EPTCS}}, \bibinfo{volume}{172}, \bibinfo{publisher}{Open
  Publishing Association}, pp. \bibinfo{pages}{270--284},
  \doi{10.4204/EPTCS.172.19}.

\bibitemdeclare{article}{zhao2019quantum}
\bibitem{zhao2019quantum}
\bibinfo{author}{Xiaobin \surnamestart Zhao\surnameend},
  \bibinfo{author}{Yuxiang \surnamestart Yang\surnameend} \&
  \bibinfo{author}{Giulio \surnamestart Chiribella\surnameend}
  (\bibinfo{year}{2020}): \emph{\bibinfo{title}{{Quantum Metrology with
  Indefinite Causal Order}}}.
\newblock {\sl \bibinfo{journal}{Physical Review Letters}}
  \bibinfo{volume}{124}(\bibinfo{number}{19}), p. \bibinfo{pages}{190503},
  \doi{10.1103/PhysRevLett.124.190503}.

\end{thebibliography}

\appendix

\section{Proof of Theorem 2}
We show that the normalisation of $\pazocal{N}_{\pi \pi'}$ is determined by the number of cycles in the cycle decomposition of $C_{\pi \pi'}$, specifically
\begin{itemize}
    \item $\pazocal{N}_{\pi \pi'} = \frac{1}{d^N}d^{c_{\pi \pi'} - 2}d \pazocal{D}$ when $0$ and $\pi(N)$ are not in the same cycle of $C_{\pi \pi'}$
    \item $\pazocal{N}_{\pi \pi'} = \frac{1}{d^N}d^{c_{\pi \pi'} - 1} \pazocal{I}$ when when $0$ and $\pi(N)$ are in the same cycle of $C_{\pi \pi'}$
\end{itemize}
\begin{proof}
Each substitution of a CDPC into a CPM-like diagram contributes a factor of $\frac{1}{d}$, substitution of all $N$ CDPCs then contributes $\frac{1}{d^N}$. The proportionality constant between $\pazocal{N}_{\pi \pi'}$ and each of the diagrams in figure $6$ is computed by counting the number of closed loops in the corresponding CPM-like diagram, each of which contributes an additional factor of $d$.
\begin{figure}[H]
    \centering
    \begin{equation}
        \begin{tikzpicture}
	\begin{pgfonlayer}{nodelayer}
		\node [style=none] (0) at (2, 0.75) {};
		\node [style=none] (1) at (3.5, 0.75) {};
		\node [style=none] (2) at (2, -0.5) {};
		\node [style=none] (3) at (3.5, -0.5) {};
		\node [style=none] (4) at (2, -0.5) {};
		\node [style=none] (5) at (3.5, -0.5) {};
		\node [style=none] (6) at (-2.5, 0.75) {};
		\node [style=none] (7) at (-1.5, 0.75) {};
		\node [style=none] (8) at (-2.5, -0.5) {};
		\node [style=none] (9) at (-1.5, -0.5) {};
		\node [style=none] (10) at (-2.5, -0.5) {};
		\node [style=none] (11) at (-1.5, -0.5) {};
		\node [style=none] (12) at (1.5, 0.25) {=};
		\node [style=none] (13) at (0.5, 0.25) {$d\pazocal{D}$};
		\node [style=none] (14) at (-3, 0.25) {=};
		\node [style=none] (15) at (-4, 0.25) {$\pazocal{I}$};
	\end{pgfonlayer}
	\begin{pgfonlayer}{edgelayer}
		\draw [bend right=90] (0.center) to (1.center);
		\draw [bend left=90] (2.center) to (3.center);
		\draw (6.center) to (10.center);
		\draw (11.center) to (7.center);
	\end{pgfonlayer}
\end{tikzpicture}
    \end{equation}
    \caption{We find the proportionality constant between $\pazocal{N}_{\pi \pi'}$ and either $\pazocal{I}$ or $d\pazocal{D}$}
    \label{diagrams}
\end{figure}
The number of cycles $c_{\pi \pi'}$ in the cycle decomposition of $C_{\pi \pi'}$ is the number of closed loops in the modified diagram for $\pazocal{N}_{\pi \pi'}$. For $\pazocal{N}_{\pi \pi'} \propto d\pazocal{D}$ the modification of the diagram has introduced $2$ new closed loops, whereas for $\pazocal{N}_{\pi \pi'} \propto \pazocal{I}$ the modification has introduced only $1$ extra closed loop into the diagram. As such the number of closed loops in the unmodified CPM-like diagram for $\pazocal{N}_{\pi \pi'}$ is
\begin{itemize}
    \item $c_{\pi \pi'} -2$ when $\pazocal{N}_{\pi \pi'} \propto d\pazocal{D}$
    \item $c_{\pi \pi'} - 1$ when $\pazocal{N}_{\pi \pi'} \propto \pazocal{I}$
\end{itemize}
and so the proportionality constant for $\pazocal{N}_{\pi \pi'}$ is given by $\frac{1}{d^N}d^{c_{\pi \pi'} - 2}$ for $\pazocal{N}_{\pi \pi'} \propto d\pazocal{D}$ and $\frac{1}{d^N}d^{c_{\pi \pi'} - 1}$ for $\pazocal{N}_{\pi \pi'} \propto \pazocal{I}$.
\end{proof}
\section{Proof of Corollary 4}
We show that the Cyclic Permutation Protocol Optimises $\mathcal{O}(S)$ for $M \leq N$
\begin{proof}
Any protocol with $M$ permutations, for which there exists $\pi,\pi'$ with $\pi' \circ \pi^{-1}$ not a cyclic permutation, will have some term $\ket{\pi}\bra{\pi'}$ with either $\pazocal{N}_{\pi \pi'}(\rho) = \alpha \frac{\rho}{d^2}$ and $\alpha < 1$ or $\pazocal{N}_{\pi \pi'}(\rho) = \beta \frac{I}{d}$. In either case \[\mathcal{O}(S) < \frac{M(M-1)\frac{1}{d^2}}{M} = \mathcal{O}(S_{\textrm{cyclic}})\]
\end{proof}

\end{document}